\documentclass[11pt]{article}

\usepackage[letterpaper,margin=1in]{geometry}
\usepackage[T1]{fontenc}
\usepackage[american]{babel}
\usepackage{lmodern}
\usepackage{microtype}

\usepackage{amsmath}
\usepackage{amssymb}
\usepackage{amsthm}
\usepackage{mathtools}
\usepackage{bm}

\usepackage{enumitem}

\usepackage{graphicx}
\usepackage{booktabs}
\usepackage{array}
\usepackage{tabularx}
\usepackage{placeins}

\graphicspath{{figures/}}

\usepackage[numbers,sort&compress]{natbib}
\usepackage[hidelinks]{hyperref}

\usepackage{authblk}

\theoremstyle{plain}
\newtheorem{theorem}{Theorem}
\newtheorem{proposition}{Proposition}
\newtheorem{corollary}{Corollary}

\newtheorem{assumption}{Assumption}

\theoremstyle{definition}
\newtheorem{definition}{Definition}

\theoremstyle{remark}

\title{Space--Time Information Interchangeability in Dynamical Systems:\\
Conditions and Bounds for Replacing Spatial Sensors with Temporal Histories}

\author[1]{Maryam Reza}
\author[2]{Farbod Faraji}

\affil[1]{Independent Researcher}
\affil[2]{Department of Computing, Imperial College London, London, UK}

\date{}

\begin{document}

\maketitle

\begin{abstract}
Many dynamical-state reconstruction problems seek to infer a high-dimensional spatial state from measurements at only a few locations. Because the governing dynamics couple evolution across space and time, temporal measurements can contain information about state components beyond the sensor locations. This work develops a theoretical framework comparing the information content of a temporal measurement history with that of a specified instantaneous spatial sensor array. The temporal history and spatial reference array are represented by prior-whitened information operators that account for prior variability, sensor geometry, and measurement noise while preserving the directional distribution of uncertainty reduction. A general space--time information-exchange theorem gives necessary and sufficient conditions for retaining a prescribed fraction of the spatial-reference information in every state direction. It answers three practical questions. First, how many temporal measurements and how much total information are required? Second, is the prescribed fraction attainable, or does a limiting information ceiling rule out every finite history? Third, when attainable, what history length is sufficient? Explicit history-length bounds are derived for contractive, spectrally separated unitary, and simultaneously diagonalizable dissipative linear dynamics, and for stochastic linear estimation under Gaussian process and measurement noise. For nonlinear systems, global delay-map injectivity is combined with a uniform local Fisher-information condition. Numerical experiments test directional exchange lengths, information ceilings, and finite-memory bounds, and assess nonlinear delay-map distinguishability and local Fisher-information predictions. The framework provides an explicit basis for determining when temporal measurements can reduce spatial sensor coverage for full-state reconstruction.
\end{abstract}


\section{Introduction}
\label{sec:introduction}

Reconstruction of a spatial state distribution from sparse measurements arises throughout plasma physics, fluid dynamics, climate modeling, and structural monitoring. After spatial discretization, the state may contain thousands or millions of variables, whereas an experimental system often provides measurements at only a small number of locations. An instantaneous inverse problem is therefore strongly underdetermined unless additional physical or statistical structure is introduced.

Dynamical evolution provides one such source of structure, making sparse temporal measurements informative about a spatially distributed state. In systems governed by partial differential equations, variations in space and time are coupled through processes such as advection, diffusion, wave propagation, and nonlinear interactions. These processes transport, disperse, deform, and reorganize spatial structures as the system evolves. A fixed sensor therefore does not repeatedly observe the same local information; instead, its temporal record reflects the passage and transformation of structures originating from different parts of the domain, as well as their coupling to other dynamically active quantities. Characteristic temporal signatures can contain information about spatial patterns, regions, and variables that are not measured directly. A sufficiently informative measurement history may therefore provide constraints comparable to those obtained from a larger simultaneous spatial sensor array. Recent recurrent-decoder data-driven architectures, such as Shallow Recurrent Decoder (SHRED), build on these ideas by learning a mapping from temporal histories to high-dimensional spatial states \cite{Williams2024SHRED,Reza2024SHRED,Faraji2025SHRED,Tomasetto2025SHRED}.

In SHRED, a recurrent encoder processes a finite history from a small number of sensors and encodes the information carried by their temporal evolution into a latent representation, which a decoder maps to the full spatial field. Its empirical success \cite{Reza2024SHRED,Faraji2025SHRED} demonstrates that temporal evolution can expose spatial information to fixed sensors, but reconstruction accuracy alone does not establish when this substitution is possible. A history may constrain some state directions well while leaving others poorly observed, and disagreement may arise either because the measurements lack sufficient information or because the decoder fails to extract information that is present. This motivates an architecture-independent theory that determines how much of the information supplied by a dense instantaneous array is retained by a sparse temporal history, whether the desired level is attainable, and how long the history must be.

This principle is closely connected to several established mathematical frameworks. Dynamical sampling studies state recovery from measurements obtained by repeatedly applying an evolution operator and sampling the resulting trajectory \cite{Aldroubi2013,Aldroubi2015,Aldroubi2017,AldroubiSurvey2026}. In this setting, temporal observations generate an orbit of effective sensing vectors, and stable recovery is characterized through frame conditions. Recent developments address noisy graph signals and sensor placement \cite{AldroubiGraph2024}, nonuniform sampling \cite{ZhangNonuniform2017}, and the design of dynamical frames \cite{AguileraOptimal2026}. Delay-coordinate reconstruction provides a complementary nonlinear perspective: Takens-type theorems establish generic conditions under which histories of partial observations embed a finite-dimensional attractor \cite{Takens1981,Sauer1991}. Stable embedding results also examine the preservation of distances and local geometry under the delay map \cite{YapRozell2011,Eftekhari2018,PanDuraisamy2020}. Observability theory supplies related rank and differential conditions for determining whether distinct states can be distinguished from their output trajectories \cite{HermannKrener1977}. Bayesian estimation theory provides the corresponding statistical description of uncertainty. Gaussian information matrices and posterior covariance formulas relate sensor geometry and measurement noise to the uncertainty remaining after observations are incorporated \cite[Sec.~10.6, pp.~325--328]{Kay1993}. Kalman filtering extends this formulation to stochastic dynamical systems \cite[Sec.~3.1, pp.~46--55]{AndersonMoore1979}. Additionally, Riccati convergence and exponential forgetting describe how rapidly the contribution of increasingly remote observations to current-state estimation decreases \cite{HagerHorowitz1976,Bougerol1993,Kozdoba2019}.

These theories establish important conditions for state recovery from measurements distributed across space and time. A question that is not directly resolved by any one framework, however, is how a temporal sensing design compares with a specified simultaneous spatial reference array. Building on these foundations, the present study quantifies the trade-off between spatial sensor coverage and temporal measurement history. Specifically, it addresses the following question:

\emph{For a specified target state, a sparse array of temporal measurements, and a denser spatial array of instantaneous measurements, what conditions guarantee that the temporal measurement history retains a prescribed fraction of the information provided by the spatial array, and what history length is required to achieve this level of information equivalence?}

The number of measurements alone is not sufficient to determine whether the target state can be identified and recovered stably. State identifiability depends on whether the observation map distinguishes all admissible system states. Stable recovery further requires this map to be sufficiently well conditioned, so that small measurement noise does not produce large reconstruction errors. In a Bayesian formulation, the comparison also depends on the prior distribution of the target state, because information gained along state-space directions with large prior variability may be more consequential than information gained along directions that are already tightly constrained.

The spatial and temporal sensing arrangements are therefore compared through a direction-resolved information operator, which quantifies the reduction in prior uncertainty achieved by the measurements along each state-space direction---that is, along a particular state variable, spatial pattern, dynamical mode, or linear combination of these---while accounting for both prior variability and measurement noise. Unlike sensor count, matrix rank, or scalar mutual information alone, this operator preserves the directional distribution of information and therefore reveals which aspects of the state are well observed and which remain weakly constrained.

Let $K_{\mathcal{T}}(L)$ represent the information about the target state contained in a temporal measurement history of length $L$, and let $K_{\mathcal{S}}$ represent the information supplied by the simultaneous spatial reference measurement array. The two sensing arrangements are compared through the operator inequality
\[
K_{\mathcal{T}}(L) \succeq \alpha K_{\mathcal{S}},
\qquad
\alpha \in (0,1].
\]

This inequality requires the temporal history to provide at least the fraction $\alpha$ of the spatial-reference information for every possible variation of the state. Equivalently, for every direction $v$,
\[
v^{\mathsf{T}} K_{\mathcal{T}}(L) v
\geq
\alpha v^{\mathsf{T}} K_{\mathcal{S}} v.
\]

When this condition holds, it yields corresponding bounds on posterior uncertainty and the expected error of the reconstructed state. The central objective is therefore to determine when this information criterion is satisfied and the minimum history length $L$ required to satisfy it. Figure~\ref{fig:space-time-sensing} schematically illustrates the two sensing strategies.

The main contributions are enumerated as follows: \textbf{(i)} a common information-operator formulation is developed for target-state estimation, in which an unknown target state is inferred from temporal measurement histories; \textbf{(ii)} exact space--time information equivalence is characterized using whitened sensing Gramians and an isometry between the corresponding measurement subspaces; \textbf{(iii)} a general exchange theorem establishes directional necessary and sufficient conditions, lower bounds based on matrix rank and total information budget, a limiting information-ceiling test, and conditions guaranteeing that the prescribed information level is reached within a finite history length; \textbf{(iv)} explicit bounds on the required history length are derived for contractive, spectrally separated unitary, and simultaneously diagonalizable dissipative linear dynamics, as well as for finite-memory state estimation in stochastic linear systems through Riccati convergence; \textbf{(v)} the framework is extended to nonlinear systems by combining global delay-map injectivity, which ensures unique state identification from measurement histories, with local Fisher-information geometry, which characterizes sensitivity to nearby state variations in the presence of noise; and \textbf{(vi)} analytical examples and numerical experiments verify the theoretical results at scalar, directional, and full-matrix levels.

To our knowledge, existing approaches do not provide a single reference-based framework that directly compares a temporal measurement history with a specified simultaneous spatial sensor array, determines whether a prescribed directional information level is attainable, and bounds the history length required to attain it.

\begin{figure}[t]
    \centering
    \includegraphics[width=\textwidth]{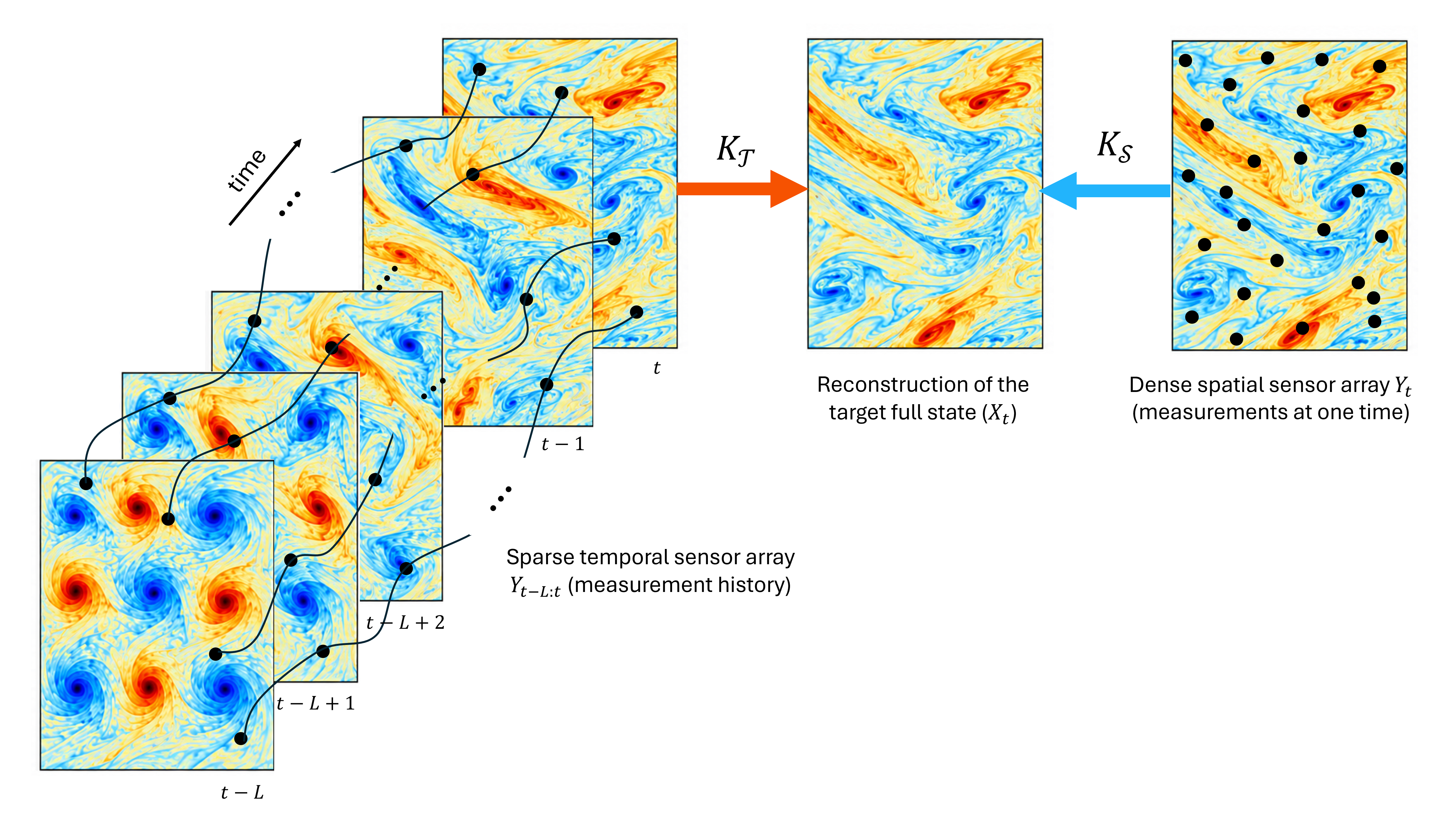}
    \caption{Schematic comparison of simultaneous spatial sensing and sparser temporal sensing for reconstruction of the same target state.}
    \label{fig:space-time-sensing}
\end{figure}

\section{Conceptual and mathematical background}
\label{sec:background}

This section introduces the principal concepts and mathematical tools used in the analysis and explains their connection to the physical problem of reconstructing a high-dimensional state from sparse measurements. It establishes the notation and intuition required for the theoretical results that follow.

\paragraph{State, target, and sensing design.}
The \emph{state} collects all variables needed to describe the system at a chosen time. A spatially discretized system such as a plasma or fluid model may contain local values of density, velocity, pressure, or electromagnetic fields at every grid point. After discretization or projection onto $r$ basis functions, the target state is represented by a vector $X \in \mathbb{R}^{r}$. The target may be the complete spatial values on a numerical grid, a reduced set of modal coefficients, or a selected collection of physically relevant quantities. When a reduced representation is used, its truncation error remains separate from the sensing and reconstruction errors considered below.

A sensing design is a mathematical description of how the state is mapped to observations. In a direct linear measurement model, the observations depend linearly on the state. A simultaneous spatial array is represented by
\begin{equation}
Y_{\mathcal{S}}
=
H_{\mathcal{S}}X+\eta_{\mathcal{S}},
\qquad
\eta_{\mathcal{S}}\sim\mathcal{N}(0,R_{\mathcal{S}}),
\label{eq:spatial-model}
\end{equation}
where $X$ is the target state to be reconstructed, $Y_{\mathcal{S}}$ is the vector of simultaneous sensor measurements, $H_{\mathcal{S}}$ is the spatial observation matrix that specifies what each sensor measures, and $\eta_{\mathcal{S}}$ represents zero-mean Gaussian measurement noise with covariance $R_{\mathcal{S}}$. Each row of $H_{\mathcal{S}}$ describes one measurement. For a point sensor, it may simply select the value of $X$ at one location. For a line-of-sight or integral measurement, it combines values from several locations according to an observation kernel. The covariance matrix $R_{\mathcal{S}}$ describes the magnitude of the measurement errors and their correlations across sensors.

A temporal design uses a smaller physical sensor matrix $C$ repeatedly as the state evolves. For deterministic discrete dynamics $X_{k+1}=AX_k$, the corresponding noiseless measurements at successive times are
\[
Y_{0}=CX_{0},
\qquad
Y_{1}=CAX_{0},
\qquad
Y_{2}=CA^{2}X_{0},
\qquad
\ldots
\]
and can be written together as
\begin{equation}
Y_{\mathcal{T},L}
=
\mathcal{O}_{L}X_{0}+\eta_{\mathcal{T}},
\qquad
\mathcal{O}_{L}
=
\begin{bmatrix}
C\\
CA\\
\vdots\\
CA^{L-1}
\end{bmatrix}.
\label{eq:temporal-stack}
\end{equation}
Here, $\eta_{\mathcal{T}}$ denotes the stacked temporal measurement noise, with covariance $R_{\mathcal{T}}$. The first block, $C$, describes how the sensors measure the initial state. The subsequent blocks, $CA,CA^{2},\ldots$, describe how the same sensors measure the state after one, two, and further steps of dynamical evolution. Although the physical sensors remain fixed, the dynamics can make them observe different combinations of the initial state over time. Thus, $m$ sensors sampled at $L$ times provide $mL$ measurement equations, called effective sensing directions, without adding physical sensors.

\paragraph{Identifiability, observability, and stable recovery.}
A state is identifiable from a measurement history when two distinct admissible states cannot produce the same noiseless observations. In linear systems, this property is called \emph{observability}. For the temporal measurement history in Eq.~\ref{eq:temporal-stack}, observability requires
\begin{equation}
\operatorname{rank}(\mathcal{O}_{L})=r,
\label{eq:observability-rank}
\end{equation}
where $\mathcal{O}_{L}$ is the observability matrix and $r$ is the dimension of the target state. Full rank means that every component or pattern of variation in the state influences the measurement history, so the state is uniquely recoverable in an ideal noise-free model.

Practical reconstruction also depends on stability. If two distinct states generate nearly identical histories, even a small amount of noise can produce a large state-estimation error. The noise-whitened observability matrix is $R_{\mathcal{T}}^{-1/2}\mathcal{O}_{L}$, where multiplication by $R_{\mathcal{T}}^{-1/2}$ rescales and decorrelates the measurements so that the transformed measurement noise has identity covariance. Its singular values then quantify how strongly different state variations appear in the measurements relative to the measurement noise. Stable recovery requires the frame inequality
\begin{equation}
a\lVert x\rVert_{2}^{2}
\leq
\left\lVert R_{\mathcal{T}}^{-1/2}\mathcal{O}_{L}x\right\rVert_{2}^{2}
\leq
b\lVert x\rVert_{2}^{2},
\qquad
0<a\leq b<\infty,
\label{eq:frame-ineq}
\end{equation}
to hold for every state variation $x$, such as the difference between two possible states. The lower frame bound $a$ measures the weakest observable state variation, with larger values indicating greater resistance to noise. The upper frame bound $b$ limits how strongly state variations can be amplified in the measurements. Dynamical sampling studies the conditions under which sensing vectors generated by repeated application of an evolution operator form such stable frames \cite{Aldroubi2017,AldroubiSurvey2026}.

\paragraph{Prior uncertainty, posterior uncertainty, and mutual information.}
A Bayesian formulation represents the target state as a random variable,
\begin{equation}
X\sim\mathcal{N}(\mu,P),
\qquad
P\succ0.
\label{eq:prior-model}
\end{equation}
The prior mean $\mu$ represents the expected state before the new measurements are used, while the covariance $P$ describes the magnitude and correlation of plausible deviations from that mean. Large variance in a particular direction indicates substantial prior uncertainty in that combination of state variables.

After observations $Y$ are incorporated, the remaining uncertainty is described by the posterior covariance
\begin{equation}
P_{\mathrm{post}}
=
\operatorname{Cov}(X\mid Y).
\label{eq:posterior-covariance}
\end{equation}

Mutual information $I(X;Y)$ measures the total amount of uncertainty about $X$ removed by observing $Y$. For jointly Gaussian variables,
\begin{equation}
I(X;Y)
=
\frac{1}{2}
\log
\frac{\det P}{\det P_{\mathrm{post}}}.
\label{eq:gaussian-mutual-information}
\end{equation}

Because mutual information summarizes the complete covariance change by a single scalar value, it does not identify how that reduction is distributed across different state directions. Direction-resolved comparisons of sensing arrangements therefore require a matrix-valued information measure.

\paragraph{The information operator and prior whitening.}
The information operator used in this paper measures the information supplied by the observations after accounting for the prior variability of the state. Prior whitening maps the centered state $X-\mu$ to
\begin{equation}
Z
=
P^{-1/2}(X-\mu),
\label{eq:prior-whitening}
\end{equation}
so that $\operatorname{Cov}(Z)=I$. Thus, in whitened coordinates, every unit variation represents one prior standard deviation, and different directions are uncorrelated. This allows information about variables with different physical units and prior variabilities to be compared on a common scale. In these coordinates, differences between sensing designs reflect the observations and measurement noise rather than unequal prior variances.

For a linear measurement model $Y=HX+\eta$ with $\eta\sim\mathcal{N}(0,R)$, the prior-whitened information operator is
\begin{equation}
K
=
P^{1/2}H^{\mathsf{T}}R^{-1}HP^{1/2}.
\label{eq:information-operator}
\end{equation}
Here, $P^{1/2}$ denotes the symmetric positive-definite square root of $P$. For a unit vector $z$ in whitened coordinates, the quadratic form $z^{\mathsf{T}}Kz$ measures the signal-to-noise-weighted information supplied about that direction. Small eigenvalues of $K$ identify combinations of the state that remain weakly constrained.

\paragraph{Positive-semidefinite ordering and directional exchange.}
For symmetric matrices $K_{1}$ and $K_{2}$, $K_{1}\succeq K_{2}$ means
\begin{equation}
x^{\mathsf{T}}K_{1}x
\geq
x^{\mathsf{T}}K_{2}x
\qquad
\text{for every }x.
\label{eq:loewner-order}
\end{equation}
This relation is known as the Loewner or positive-semidefinite order. It compares two operators in every direction simultaneously. The condition
\begin{equation}
K_{\mathcal{T}}(L)
\succeq
\alpha K_{\mathcal{S}},
\label{eq:directional-exchange}
\end{equation}
introduced previously in Section~\ref{sec:introduction} therefore states that the temporal history supplies at least the fraction $\alpha$ of the spatial-reference information in every state direction measured by the reference. The largest factor for which Eq.~\ref{eq:directional-exchange} holds is the directional substitution factor $\alpha_{L}$. Accordingly:

\begin{itemize}
    \item $\alpha_{L}<1$: the temporal history supplies less information than the spatial reference in its weakest relative direction.

    \item $\alpha_{L}=1$: the temporal information equals the spatial reference in its weakest relative direction and equals or exceeds it in all other directions.

    \item $\alpha_{L}>1$: the temporal history is more informative than the spatial reference even in its weakest relative direction.
\end{itemize}

The matrix inequality in Eq.~\ref{eq:directional-exchange} is the central definition of partial space--time exchange used throughout the paper.

\paragraph{Filtering covariance in stochastic linear systems.}
For the stochastic linear system
\begin{equation}
\begin{aligned}
X_{k+1} &= AX_{k}+W_{k},
&\qquad
W_{k} &\sim \mathcal{N}(0,Q),\\
Y_{k} &= CX_{k}+V_{k},
&\qquad
V_{k} &\sim \mathcal{N}(0,R),
\end{aligned}
\label{eq:stochastic-linear-model}
\end{equation}
where $\{W_k\}$ and $\{V_k\}$ are independent, temporally white Gaussian process- and measurement-noise sequences with covariances $Q$ and $R$, respectively. The filtering covariance
\begin{equation}
P_{k\mid k}
=
\operatorname{Cov}\!\left(X_{k}\mid Y_{0:k}\right),
\label{eq:filtering-covariance}
\end{equation}
describes the uncertainty remaining in the estimate of $X_{k}$ after all observations up to time $k$ have been incorporated. Its prediction and measurement-update steps are
\begin{subequations}
\label{eq:riccati-recursion}
\begin{align}
P_{k+1\mid k}
&=
AP_{k\mid k}A^{\mathsf{T}}+Q,
\label{eq:riccati-prediction}
\\
P_{k+1\mid k+1}
&=
P_{k+1\mid k}
\nonumber\\
&\quad
-
P_{k+1\mid k}C^{\mathsf{T}}
\left(
CP_{k+1\mid k}C^{\mathsf{T}}+R
\right)^{-1}
CP_{k+1\mid k}.
\label{eq:riccati-update}
\end{align}
\end{subequations}

Together, Eqs.~\ref{eq:riccati-prediction} and~\ref{eq:riccati-update} form the discrete Riccati recursion. Under suitable stabilizability and detectability conditions, the filtering covariance converges to a limiting value, and the influence of increasingly remote observations on the current estimate decays.

A standard sufficient condition for convergence to a limiting filtering covariance is that $(A,Q^{1/2})$ is stabilizable and $(A,C)$ is detectable \cite{HagerHorowitz1976,Bougerol1993}, where $Q^{1/2}$ denotes a factor of the process-noise covariance satisfying
\begin{equation}
Q^{1/2}\left(Q^{1/2}\right)^{\mathsf{T}}
=
Q.
\label{eq:process-noise-factor}
\end{equation}

Stabilizability requires every unstable or non-decaying mode of $A$ to lie in the controllable subspace generated by $Q^{1/2}$, while detectability requires every such mode to influence the measurements. Equivalently, for every eigenvalue $\lambda$ of $A$ with $|\lambda|\geq1$,
\begin{subequations}
\label{eq:stabilizability-detectability}
\begin{align}
\operatorname{rank}
\begin{bmatrix}
\lambda I-A & Q^{1/2}
\end{bmatrix}
&=r,
\label{eq:stabilizability-rank}
\\
\operatorname{rank}
\begin{bmatrix}
\lambda I-A\\
C
\end{bmatrix}
&=r.
\label{eq:detectability-rank}
\end{align}
\end{subequations}
Thus, unstable or non-decaying modes cannot remain uncontrollable through the process-noise directions or invisible to the sensors.

\paragraph{Nonlinear observability: delay-map injectivity and Fisher information.}
In nonlinear systems, the dynamically admissible states often occupy a lower-dimensional set $\mathcal{M}$ within the ambient state space. This set may represent an invariant manifold, an attractor, or a regime-restricted collection of physically admissible states. For dynamics
$F:\mathcal{M}\rightarrow\mathcal{M}$
and a temporal observation function
$h:\mathcal{M}\rightarrow\mathbb{R}^{m}$,
the length-$L$ delay map is
\begin{equation}
\Phi_{\mathcal{T},L}(x)
=
\bigl(
h(x),
h(Fx),
\ldots,
h(F^{L-1}x)
\bigr).
\label{eq:nonlinear-delay-map}
\end{equation}

Injectivity of $\Phi_{\mathcal{T},L}$ means that distinct admissible states generate distinct noiseless histories and therefore provides global identifiability. An embedding is stronger: it is injective and preserves the local differential structure of $\mathcal{M}$. Takens-type results give generic conditions under which sufficiently many delayed observations form such an embedding, although these results do not by themselves quantify robustness to measurement noise \cite{Takens1981,Sauer1991}.

Local robustness is determined by the differential of the observation map. If $\mathcal{M}$ is a differentiable manifold, the relevant derivative is the restriction
\begin{equation}
D\Phi_{\mathcal{T},L}(x)
\big|_{T_{x}\mathcal{M}},
\label{eq:tangent-delay-jacobian}
\end{equation}
because $T_{x}\mathcal{M}$ contains the locally admissible state perturbations. With additive Gaussian noise of covariance $R_{\mathcal{T},L}$, the corresponding Fisher-information form on the tangent space is
\begin{equation}
\mathfrak{g}_{\mathcal{T},L}(x)
=
D\Phi_{\mathcal{T},L}(x)^{\mathsf{T}}
R_{\mathcal{T},L}^{-1}
D\Phi_{\mathcal{T},L}(x),
\label{eq:temporal-fisher-information}
\end{equation}
where the derivatives in Eq.~\ref{eq:temporal-fisher-information} are understood as restricted to $T_x\mathcal{M}$ as in Eq.~\ref{eq:tangent-delay-jacobian}. Positive definiteness on $T_{x}\mathcal{M}$ gives local differential observability, while a uniform lower bound over $x\in\mathcal{M}$ provides protection against locally ill-conditioned regions.

Global injectivity and local Fisher information therefore address complementary aspects of nonlinear reconstruction. Injectivity rules out widely separated states that produce identical noiseless histories, whereas Fisher information controls the local sensitivity of the observations to small perturbations in the presence of noise. Under the usual differentiability, regularity, and local-unbiasedness assumptions, the inverse Fisher-information matrix gives a local Cram\'er--Rao lower bound. It does not, by itself, guarantee global reconstruction or characterize finite-noise errors far from the local asymptotic regime.

The core notation introduced in this section is summarized in Table~\ref{tab:notation}.

\begin{table}[t]
\centering
\caption{Core quantities used throughout the analysis.}
\label{tab:notation}
\begin{tabularx}{\linewidth}{@{}lX@{}}
\toprule
Symbol & Meaning \\
\midrule
$X\in\mathbb{R}^{r}$
& Target $r$-dimensional state, either full or reduced \\

$H_{\mathcal{S}}$
& Simultaneous spatial observation matrix \\

$C$
& Sparse physical sensor matrix used through time \\

$L$
& Number of temporal samples in the history \\

$P$
& Prior covariance of the target state \\

$R_{\mathcal{S}},R_{\mathcal{T}}$
& Spatial and temporal measurement-noise covariances \\

$K_{\mathcal{S}},K_{\mathcal{T}}(L)$
& Prior-whitened information operators \\

$\alpha$
& Required fraction of the spatial-reference information \\

$\alpha_{L}$
& Directional substitution factor attained with a temporal history of length $L$ \\

$L_{\alpha}$
& Minimum history length satisfying
$K_{\mathcal{T}}(L)\succeq\alpha K_{\mathcal{S}}$ \\
\bottomrule
\end{tabularx}
\end{table}

\section{Information-operator foundations for sensing equivalence}
\label{sec:information-operators}

This section establishes a common mathematical basis for comparing spatial and temporal sensing designs. The information operator converts the posterior uncertainty produced by each design into the same prior-normalized coordinates, allowing exact and partial sensing equivalence to be defined direction by direction.

Let $X\sim\mathcal{N}(\mu,P)$ with $P\succ0$, and let $Y_a$ denote the
observations produced by design
$a\in\{\mathcal{S},\mathcal{T}\}$,
corresponding to the spatial or temporal sensing arrangement. Assume that
$(X,Y_a)$ is jointly Gaussian. Specializing the generic posterior covariance
in Eq.~\ref{eq:posterior-covariance} to design $a$, define
\begin{equation}
P_a
:=
\operatorname{Cov}(X\mid Y_a)
\succ0.
\label{eq:design-posterior-covariance}
\end{equation}

The posterior uncertainty is first normalized by the prior covariance, and the corresponding information operator is then defined by
\begin{subequations}
\label{eq:normalized-information-definitions}
\begin{align}
\Pi_a
&=
P^{-1/2}P_aP^{-1/2},
\label{eq:normalized-posterior-covariance}
\\
K_a
&=
\Pi_a^{-1}-I.
\label{eq:design-information-operator}
\end{align}
\end{subequations}
Because conditioning on observations cannot increase covariance in the jointly Gaussian setting, $\Pi_a\preceq I$, and therefore $K_a\succeq0$. Thus, $K_a$ represents the information gained from design $a$, expressed relative to the prior uncertainty.

As stated in Section~\ref{sec:background}, for a direct linear-Gaussian experiment
$Y_a=H_aX+\eta_a$
with measurement noise $\eta_a\sim\mathcal{N}(0,R_a)$ independent of $X$, this general definition reduces to
\begin{equation}
K_a
=
P^{1/2}H_a^{\mathsf{T}}R_a^{-1}H_aP^{1/2}.
\label{eq:direct-design-information-operator}
\end{equation}

\begin{proposition}[Posterior and mutual information]
\label{prop:posterior-mi}
For any jointly Gaussian target and observation pair,
\begin{subequations}
\label{eq:posterior-mi-identities}
\begin{align}
P_a
&=
P^{1/2}(I+K_a)^{-1}P^{1/2},
\label{eq:posterior-from-information}
\\
I(X;Y_a)
&=
\frac{1}{2}\log\det(I+K_a).
\label{eq:mutual-information-from-operator}
\end{align}
\end{subequations}
\end{proposition}

\begin{proof}
The first identity follows by rearranging
Eq.~\ref{eq:design-information-operator} and using
Eq.~\ref{eq:normalized-posterior-covariance}. The Gaussian entropy
formula gives
$I(X;Y_a)=\frac{1}{2}\log\det P-\frac{1}{2}\log\det P_a$,
and substitution of Eq.~\ref{eq:posterior-from-information} yields
Eq.~\ref{eq:mutual-information-from-operator}.
\end{proof}

The following definition formalizes the direction-resolved comparison of the information supplied by the spatial and temporal sensing designs. Exact equivalence requires the two designs to produce the same information structure. Partial exchange requires the temporal history to retain a prescribed fraction of the spatial information in every direction.

\begin{definition}[Exact and one-sided information exchange]
\label{def:information-exchange}
The designs are exactly information-equivalent if
\begin{equation}
K_{\mathcal{T}}
=
K_{\mathcal{S}}.
\label{eq:exact-information-equivalence}
\end{equation}
For $\alpha\in(0,1]$, the temporal design is an $\alpha$-substitute for
the spatial design if
\begin{equation}
K_{\mathcal{T}}
\succeq
\alpha K_{\mathcal{S}}.
\label{eq:alpha-substitute}
\end{equation}
The largest admissible fraction is
\begin{equation}
\alpha_{\mathrm{op}}(\mathcal{T}\mid\mathcal{S})
=
\sup
\left\{
\beta\geq0:
K_{\mathcal{T}}\succeq\beta K_{\mathcal{S}}
\right\}.
\label{eq:operator-substitution-factor}
\end{equation}
\end{definition}

The supremum in Eq.~\ref{eq:operator-substitution-factor} may exceed
one; the restriction $\alpha\in(0,1]$ applies to the prescribed fraction
used to define an $\alpha$-substitute. When the temporal design is the
length-$L$ measurement history considered in Section~\ref{sec:background},
denote it by $\mathcal{T}_L$. The directional substitution factor
$\alpha_L$ introduced there is therefore
\begin{equation}
\alpha_L
=
\alpha_{\mathrm{op}}(\mathcal{T}_L\mid\mathcal{S}).
\label{eq:alpha-L-substitution-factor}
\end{equation}

For direct linear experiments, define the whitened sensing matrix
\begin{equation}
B_a
=
R_a^{-1/2}H_aP^{1/2},
\label{eq:whitened-sensing-matrix}
\end{equation}
where $R_a^{-1/2}$ denotes the inverse of the symmetric
positive-definite square root of the measurement-noise covariance
$R_a$. The matrix $B_a$ maps perturbations in prior-whitened state
coordinates to noise-whitened measurement coordinates. By
Eq.~\ref{eq:direct-design-information-operator},
\begin{equation}
K_a
=
B_a^{\mathsf{T}}B_a.
\label{eq:information-gram}
\end{equation}
Thus, $K_a$ is the Gram matrix of the columns of $B_a$, the whitened
sensing Gramian associated with design $a$.

\begin{theorem}[Exact Gram and measurement-subspace characterization]
\label{thm:gram-characterization}
For two direct linear-Gaussian sensing designs, the following are
equivalent:
\begin{enumerate}[label=(\roman*)]
    \item
    the designs are exactly information-equivalent,
    $K_{\mathcal{T}}=K_{\mathcal{S}}$;

    \item
    the whitened sensing Gram matrices are equal,
    \[
    B_{\mathcal{T}}^{\mathsf{T}}B_{\mathcal{T}}
    =
    B_{\mathcal{S}}^{\mathsf{T}}B_{\mathcal{S}};
    \]

    \item
    there exists a partial isometry $U$ from the spatial measurement
    space to the temporal measurement space whose restriction to
    $\operatorname{range}(B_{\mathcal{S}})$ is an isometry, such that
    \begin{equation}
    B_{\mathcal{T}}
    =
    UB_{\mathcal{S}}.
    \label{eq:partial-isometry-relation}
    \end{equation}
\end{enumerate}

The partial isometry $U$ shows that the two designs may record different
measurement vectors while preserving the same inner products between
all observable state perturbations. Exact exchange is therefore a
statement about the geometry of the inferred state rather than a
one-to-one correspondence between individual sensors.
\end{theorem}

\begin{proof}
The equivalence of (i) and (ii) follows directly from
Eq.~\ref{eq:information-gram}. Assume (ii) and define $U$ on
$\operatorname{range}(B_{\mathcal{S}})$ by
$U(B_{\mathcal{S}}x)=B_{\mathcal{T}}x$. This is well defined because
$B_{\mathcal{S}}x=0$ implies, by (ii),
$\lVert B_{\mathcal{T}}x\rVert_2^2
=
x^{\mathsf{T}}B_{\mathcal{T}}^{\mathsf{T}}B_{\mathcal{T}}x
=
x^{\mathsf{T}}B_{\mathcal{S}}^{\mathsf{T}}B_{\mathcal{S}}x
=0$,
and hence $B_{\mathcal{T}}x=0$. Moreover,
\begin{equation}
\begin{aligned}
\left\langle B_{\mathcal{S}}x,B_{\mathcal{S}}z\right\rangle
&=
x^{\mathsf{T}}
B_{\mathcal{S}}^{\mathsf{T}}B_{\mathcal{S}}z
\\
&=
x^{\mathsf{T}}
B_{\mathcal{T}}^{\mathsf{T}}B_{\mathcal{T}}z
\\
&=
\left\langle B_{\mathcal{T}}x,B_{\mathcal{T}}z\right\rangle,
\end{aligned}
\label{eq:measurement-inner-product-preservation}
\end{equation}
so $U$ is an isometry on
$\operatorname{range}(B_{\mathcal{S}})$. Extending it by zero on the
orthogonal complement gives a partial isometry satisfying (iii).

Conversely, under (iii), $U^{\mathsf{T}}U$ acts as the identity on
$\operatorname{range}(B_{\mathcal{S}})$. Therefore,
\begin{equation}
B_{\mathcal{T}}^{\mathsf{T}}B_{\mathcal{T}}
=
B_{\mathcal{S}}^{\mathsf{T}}
U^{\mathsf{T}}U
B_{\mathcal{S}}
=
B_{\mathcal{S}}^{\mathsf{T}}B_{\mathcal{S}},
\label{eq:partial-isometry-gram}
\end{equation}
which is precisely the equality of the spatial and temporal Gram
matrices in (ii).
\end{proof}

\begin{theorem}[Exact fixed-prior Bayesian equivalence]
\label{thm:bayes-equivalence}
If $K_{\mathcal{T}}=K_{\mathcal{S}}$ for two jointly Gaussian experiments
with the same prior, then they have identical posterior covariance,
identical mutual information, identical minimum mean-square error, and
identical minimum Bayes risk for every integrable loss under that fixed
prior.
\end{theorem}

\begin{proof}
Equality of $K_{\mathcal{T}}$ and $K_{\mathcal{S}}$ gives
$P_{\mathcal{T}}=P_{\mathcal{S}}$ and equal mutual information by
Proposition~\ref{prop:posterior-mi}. The common posterior covariance also
gives the same minimum mean-square error. For each
$a\in\{\mathcal{S},\mathcal{T}\}$, the posterior distribution is
$\mathcal{N}(M_a,P_a)$, where
$M_a=\mathbb{E}[X\mid Y_a]$. Since $\mathbb{E}[M_a]=\mu$ and
$\operatorname{Cov}(M_a)=\operatorname{Cov}(X,M_a)=P-P_a$, the equality
$P_{\mathcal{T}}=P_{\mathcal{S}}$ implies that
$(X,M_{\mathcal{T}})$ and $(X,M_{\mathcal{S}})$ have the same joint
Gaussian distribution. The posterior expected loss is therefore
distributed identically under the spatial and temporal experiments, so
its optimum has the same Bayes risk.
\end{proof}

The equivalence in Theorem~\ref{thm:bayes-equivalence} is defined for
the specified prior $P$. A prior-independent ordering of statistical
experiments is provided by Blackwell's comparison theory
\cite{Blackwell1953}; the fixed-prior formulation used here is tailored
to reconstruction of a prescribed ensemble of dynamical states.

Mutual information remains useful for summarizing total uncertainty
reduction, although equal totals can arise from different directional
distributions. The following elementary example illustrates why the
operator comparison is stronger.

\begin{proposition}[Scalar mutual information is insufficient]
\label{prop:scalar-mutual-information}
For $r\geq2$, there exist positive-semidefinite $K_1,K_2$ with
\begin{equation}
\log\det(I+K_1)
=
\log\det(I+K_2),
\label{eq:equal-scalar-mutual-information}
\end{equation}
but neither $K_1\succeq K_2$ nor $K_2\succeq K_1$.
\end{proposition}

\begin{proof}
Take
$K_1=\operatorname{diag}((1+b)^2-1,0,0,\ldots)$ and
$K_2=\operatorname{diag}(b,b,0,\ldots)$ with $b>0$, with any remaining
diagonal entries equal to zero when $r>2$. Both matrices are
positive semidefinite. Moreover,
$\det(I+K_1)=(1+b)^2=\det(I+K_2)$, while $K_1-K_2$ has first and second
diagonal entries $b+b^2>0$ and $-b<0$, respectively. Hence neither
$K_1\succeq K_2$ nor $K_2\succeq K_1$.
\end{proof}

This formulation establishes that exact sensing equivalence is
characterized by equality of the information operators, or equivalently, for direct linear-Gaussian designs, by equality of their prior- and noise-whitened Gram matrices. Under a fixed prior, this equality implies identical posterior covariance and optimal Bayesian reconstruction performance, even when the measurements themselves differ. Equal mutual information alone does not provide the same guarantee because it does not preserve the directional distribution of information. The following sections build on this foundation.

\section{General space--time exchange theorem}
\label{sec:general-theorem}

The space--time information exchange analysis is organized around three
quantities: the substitution factor $\alpha_L$, which gives the weakest
temporal-to-spatial information ratio across state-space directions; the
limiting temporal information operator $K_{\infty}$; and the convergence
rate of $K_{\mathcal{T}}(L)$ toward this limit. Together, these yield
directional necessary and sufficient conditions, rank- and trace-based
lower bounds, information-ceiling impossibility criteria, and
finite-history guarantees for the exchange length.

Histories are assumed to be nested, so the observations in a window of
length $L$ are contained in a window of length $L+1$. In a jointly
Gaussian problem, adding an observation can only reduce posterior
covariance. Consequently, the temporal information operator grows
monotonically with $L$:
\begin{equation}
K_{\mathcal{T}}(L+1)
\succeq
K_{\mathcal{T}}(L).
\label{eq:monotone-K}
\end{equation}

Define the operator exchange length as the smallest temporal history
length $L$ that supplies at least the fraction $\alpha$ of the
spatial-reference information in every direction:
\begin{equation}
L_{\alpha}(\mathcal{T}\mid\mathcal{S})
=
\inf
\left\{
L:
K_{\mathcal{T}}(L)
\succeq
\alpha K_{\mathcal{S}}
\right\}.
\label{eq:exchange-length}
\end{equation}
When the spatial and temporal designs are fixed, we write this quantity
simply as $L_{\alpha}$. If no finite $L$ satisfies the inequality, then
$L_{\alpha}=\infty$.

\begin{theorem}[General space--time exchange theorem]
\label{thm:unified-exchange}
Let $K_{\mathcal{S}}\succeq0$, let
$K_L:=K_{\mathcal{T}}(L)\succeq0$ be nondecreasing, and suppose that
$K_L$ converges in operator norm as $L\rightarrow\infty$, with limiting
value
$K_{\infty}=\lim_{L\rightarrow\infty}K_L$.
Here, $K_{\infty}$ represents the maximum information obtainable from an
arbitrarily long temporal history. Fix $\alpha\in(0,1]$. If
$K_{\mathcal{S}}=0$, adopt the convention $\alpha_L=+\infty$.

\begin{enumerate}[label=(\roman*)]

\item
\textbf{\textup{Directional substitution factor.}}
The exact substitution factor is
\begin{equation}
\alpha_L
=
\inf_{x:\,x^{\mathsf{T}}K_{\mathcal{S}}x>0}
\frac{x^{\mathsf{T}}K_Lx}
{x^{\mathsf{T}}K_{\mathcal{S}}x},
\label{eq:rayleigh-factor}
\end{equation}
and $L_{\alpha}<\infty$ exactly when $\alpha_L\geq\alpha$ for some
finite $L$.

{\normalfont
The quotient compares temporal and spatial information along the state
variation $x$. Taking the infimum identifies the weakest direction, so
$\alpha_L$ is the largest information fraction guaranteed simultaneously
across all directions represented by the spatial reference.
\par}

\item
\textbf{\textup{Rank and measurement-count requirements.}}
If
$K_L\succeq\alpha K_{\mathcal{S}}$, then
\begin{equation}
\operatorname{rank}(K_L)
\geq
\operatorname{rank}(K_{\mathcal{S}}).
\label{eq:rank-lower}
\end{equation}
If $K_L$ is generated by $mL$ scalar measurements, then necessarily
\begin{equation}
mL
\geq
\operatorname{rank}(K_{\mathcal{S}}).
\label{eq:count-lower}
\end{equation}

{\normalfont
Therefore, the temporal history must span at least as many informative
state directions as the spatial reference. This gives a necessary lower
bound on the total number of scalar temporal observations, although
satisfying this count alone does not guarantee exchange.
\par}

\item
\textbf{\textup{Information-ceiling impossibility condition.}}
If there exists a state variation $v$ such that
\begin{equation}
v^{\mathsf{T}}K_{\infty}v
<
\alpha v^{\mathsf{T}}K_{\mathcal{S}}v,
\label{eq:no-go-direction}
\end{equation}
then $L_{\alpha}=\infty$.

{\normalfont
In this case, the temporal sensing design remains below the required
exchange level for every history length. Even the complete temporal
history contains insufficient information in direction $v$.
\par}

\item
\textbf{\textup{Finite-history guarantee.}}
Suppose that the limiting temporal information exceeds the required
spatial-reference level by a uniform positive margin,
\begin{equation}
\Delta_{\infty}
:=
K_{\infty}
-
\alpha K_{\mathcal{S}}
\succeq
\gamma I,
\qquad
\gamma>0,
\label{eq:gap}
\end{equation}
and suppose the finite-history error satisfies
\begin{equation}
\left\lVert
K_{\infty}-K_L
\right\rVert_2
\leq
\phi(L),
\qquad
\phi(L)\downarrow0.
\label{eq:envelope}
\end{equation}
Then every $L$ satisfying $\phi(L)\leq\gamma$ guarantees
\begin{equation}
K_L
\succeq
\alpha K_{\mathcal{S}}.
\label{eq:finite-guarantee}
\end{equation}

{\normalfont
Here, $\gamma$ measures the smallest directional margin by which the
limiting temporal information exceeds the required level, while
$\phi(L)$ measures how far the finite-history operator remains from its
limit. Exchange is guaranteed once the remaining approximation error is
smaller than this margin.
\par}

\item
\textbf{\textup{Total information-budget requirement.}}
If
\begin{equation}
K_L
=
\sum_{k=1}^{L}\Delta K_k,
\qquad
\Delta K_k\succeq0,
\qquad
\operatorname{tr}(\Delta K_k)\leq\tau,
\qquad
\tau>0,
\label{eq:information-increment-bound}
\end{equation}
then substitution requires
\begin{equation}
L
\geq
\frac{
\alpha\,\operatorname{tr}(K_{\mathcal{S}})
}{
\tau
}.
\label{eq:trace-lower}
\end{equation}

{\normalfont
This gives a second necessary lower bound. If each additional temporal
sample contributes at most $\tau$ units of total information,
sufficiently many samples are required to accumulate the total
information demanded by the spatial reference.
\par}

\end{enumerate}
\end{theorem}

\begin{proof}
For (i), for any $\beta\geq0$, $K_L\succeq\beta K_{\mathcal{S}}$ is equivalent to
$x^{\mathsf{T}}K_Lx
\geq
\beta x^{\mathsf{T}}K_{\mathcal{S}}x$
for every $x$, and the largest such $\beta$ is the generalized Rayleigh infimum in Eq.~\ref{eq:rayleigh-factor}.

For (ii), if $x\in\ker(K_L)$, then
$0=x^{\mathsf{T}}K_Lx
\geq
\alpha x^{\mathsf{T}}K_{\mathcal{S}}x$,
so $x\in\ker(K_{\mathcal{S}})$. Hence
$\ker(K_L)\subseteq\ker(K_{\mathcal{S}})$ and the rank inequality
follows. A matrix generated by $mL$ scalar measurements has rank at most
$mL$.

For (iii), monotonicity in Eq.~\ref{eq:monotone-K} and convergence imply
$v^{\mathsf{T}}K_Lv
\leq
v^{\mathsf{T}}K_{\infty}v$
for all $L$, contradicting substitution in direction $v$.

For (iv),
\[
K_L-\alpha K_{\mathcal{S}}
=
\Delta_{\infty}
-
(K_{\infty}-K_L)
\succeq
\left[
\gamma
-
\left\lVert K_{\infty}-K_L\right\rVert_2
\right]I,
\]
which is positive semidefinite when $\phi(L)\leq\gamma$.

Finally, taking traces in
$K_L\succeq\alpha K_{\mathcal{S}}$
and using the increment bound in
Eq.~\ref{eq:information-increment-bound} proves (v).
\end{proof}

The theorem separates the exchange problem into three interpretable
components. The generalized Rayleigh quotient identifies the weakest
direction relative to the spatial reference. The limiting operator
$K_{\infty}$ determines whether the desired exchange level is attainable
at all. Once a positive limiting gap is present, the convergence
envelope $\phi(L)$ converts asymptotic information into a finite
history-length guarantee. The remainder of the paper derives these
quantities for specific classes of dynamics.

\begin{corollary}[Posterior and information consequences]
\label{cor:posterior-consequence}
For jointly Gaussian experiments and $\alpha\in(0,1]$, define
$P_{\mathcal{T},L}:=\operatorname{Cov}(X\mid Y_{\mathcal{T},L})$.
If
\begin{equation}
K_{\mathcal{T}}(L)
\succeq
\alpha K_{\mathcal{S}},
\label{eq:corollary-exchange-condition}
\end{equation}
then
\begin{equation}
P_{\mathcal{T},L}
\preceq
P^{1/2}
\left(
I+\alpha K_{\mathcal{S}}
\right)^{-1}
P^{1/2},
\label{eq:posterior-exchange-bound}
\end{equation}
and
\begin{equation}
I\left(X;Y_{\mathcal{T},L}\right)
\geq
\frac{1}{2}
\log\det
\left(
I+\alpha K_{\mathcal{S}}
\right).
\label{eq:mutual-information-exchange-bound}
\end{equation}

{\normalfont
The first inequality bounds the posterior covariance of the temporal
design: its remaining uncertainty cannot exceed the stated
prior-normalized bound. The second gives a lower bound on the total
mutual information supplied by the temporal history.
\par}

Under squared Euclidean loss, the posterior mean is the Bayes-optimal
estimator, and its minimum expected squared reconstruction error is equal
to the trace of the posterior covariance. Therefore, the temporal design
satisfies the conservative bound
\begin{equation}
\operatorname{tr}\left(P_{\mathcal{T},L}\right)
\leq
\frac{1}{\alpha}
\operatorname{tr}\left(P_{\mathcal{S}}\right).
\label{eq:trace-risk-bound}
\end{equation}
\end{corollary}

\begin{proof}
From Eq.~\ref{eq:corollary-exchange-condition},
\[
I+K_{\mathcal{T}}(L)
\succeq
I+\alpha K_{\mathcal{S}}.
\]
Inversion reverses the Loewner order, so Proposition~\ref{prop:posterior-mi}
gives the posterior-covariance bound in
Eq.~\ref{eq:posterior-exchange-bound}. The mutual-information bound in
Eq.~\ref{eq:mutual-information-exchange-bound} follows from
Eq.~\ref{eq:mutual-information-from-operator} and the monotonicity of
$\log\det$ on the positive-definite cone.

The reconstruction-error bound follows from
\[
\left(
I+\alpha K_{\mathcal{S}}
\right)^{-1}
\preceq
\frac{1}{\alpha}
\left(
I+K_{\mathcal{S}}
\right)^{-1},
\qquad
0<\alpha\leq1.
\]
Premultiplying and postmultiplying by $P^{1/2}$, using
Proposition~\ref{prop:posterior-mi} for the spatial design, and taking
traces gives Eq.~\ref{eq:trace-risk-bound}.
\end{proof}

\section{Application to deterministic linear dynamics:
explicit exchange conditions and history-length bounds}
\label{sec:linear-dynamics}

We apply here the general information-exchange framework
developed in Sections~\ref{sec:information-operators}
and~\ref{sec:general-theorem} to deterministic linear dynamical systems.
The temporal information operator is written explicitly in terms of the
linear evolution and sensor matrices, allowing concrete exchange
conditions and history-length bounds to be derived for several important
classes of linear dynamics.

For deterministic linear dynamics, the information accumulated over a
finite measurement history can be written explicitly in terms of the
evolution matrix $A$ and the physical sensor matrix $C$. Consider a
linear system $X_{k+1}=AX_k$ for which the effective observation matrix
at time $k$ is $H_k=CA^k$, with independent measurement noise of
covariance $R\succ0$ at each observation time. Applying the
information-operator expression in
Eq.~\ref{eq:direct-design-information-operator}, the accumulated
information from the $L$ independent measurements is
\begin{equation}
K_L
=
P^{1/2}
\left[
\sum_{k=0}^{L-1}
(A^k)^{\mathsf{T}}
C^{\mathsf{T}}R^{-1}CA^k
\right]
P^{1/2},
\label{eq:linear-KL}
\end{equation}
where, for brevity in this section,
$K_L:=K_{\mathcal{T}}(L)$.

Equivalently, using the finite-horizon observability matrix introduced
in Eq.~\ref{eq:temporal-stack},
\[
\mathcal{O}_L
=
\begin{bmatrix}
C\\
CA\\
\vdots\\
CA^{L-1}
\end{bmatrix},
\]
the temporal information operator can be written as
\begin{equation}
K_L
=
P^{1/2}
\mathcal{O}_L^{\mathsf{T}}
\left(
I_L\otimes R^{-1}
\right)
\mathcal{O}_L
P^{1/2},
\label{eq:linear-KL-observability}
\end{equation}
where $I_L$ is the $L\times L$ identity matrix and $\otimes$ denotes the
Kronecker product.

\begin{proposition}[Observability and positive-definite temporal information]
\label{prop:observability-positive-information}
If $P\succ0$ and $R\succ0$, then $K_L\succ0$ if and only if the
finite-horizon observability matrix $\mathcal{O}_L$ has full column rank.
\end{proposition}

\begin{proof}
From Eq.~\ref{eq:linear-KL-observability}, for every
$x\in\mathbb{R}^r$,
\[
x^{\mathsf{T}}K_Lx
=
\left\lVert
\left(I_L\otimes R^{-1/2}\right)
\mathcal{O}_L P^{1/2}x
\right\rVert_2^2.
\]
Because $P\succ0$, $P^{1/2}$ is invertible and
$P^{1/2}x\neq0$ whenever $x\neq0$. Likewise, $R\succ0$ implies that
$I_L\otimes R^{-1/2}$ is invertible. Therefore, the quadratic form is
positive for every nonzero $x$ if and only if $\mathcal{O}_L$ has
trivial null space, equivalently full column rank.
\end{proof}

Full column rank of $\mathcal{O}_L$ means that every initial-state
direction influences the finite measurement history. The temporal
information operator is therefore positive definite exactly when no nonzero state direction remains unobserved.

\subsection{Contractive dynamics}
\label{subsec:contractive}

Contractive dynamics describe systems for which the effect of the initial state decreases with time, as can occur in strictly dissipative diffusion or diffusion–reaction systems. When
\[
\lVert A\rVert_2\leq\rho<1,
\]
successive terms in the temporal information operator become
progressively smaller. Consequently, the accumulated information
approaches a finite limiting value rather than increasing without bound.
This subsection quantifies the remaining information beyond a history of
length $L$ and uses this tail estimate to derive a sufficient exchange
length.

Let
\begin{equation}
G
=
C^{\mathsf{T}}R^{-1}C.
\label{eq:sensor-information-G}
\end{equation}

\begin{theorem}[Explicit exchange length for contractive deterministic dynamics]
\label{thm:stable-tail}
Assume $\lVert A\rVert_2\leq\rho<1$. Then the limiting temporal
information operator
\begin{equation}
K_{\infty}
=
P^{1/2}
\left[
\sum_{k=0}^{\infty}
(A^k)^{\mathsf{T}}GA^k
\right]
P^{1/2}
\label{eq:contractive-K-infinity}
\end{equation}
exists, and the finite-history information deficit satisfies
\begin{equation}
\lVert K_{\infty}-K_L\rVert_2
\leq
\frac{
\lVert P\rVert_2
\lVert G\rVert_2
\rho^{2L}
}{
1-\rho^2
}.
\label{eq:stable-tail}
\end{equation}

Fix $\alpha\in(0,1]$. If
\begin{equation}
\gamma
=
\lambda_{\min}
\left(
K_{\infty}-\alpha K_{\mathcal{S}}
\right)
>0,
\label{eq:contractive-gap}
\end{equation}
then the sufficient history length
\begin{equation}
L
\geq
\left\lceil
\frac{
\log\left(
\frac{
\lVert P\rVert_2\lVert G\rVert_2
}{
\gamma(1-\rho^2)
}
\right)
}{
2\log(1/\rho)
}
\right\rceil_{+},
\label{eq:stable-L}
\end{equation}
guarantees
$K_L\succeq\alpha K_{\mathcal{S}}$.
Here,
$\lceil z\rceil_{+}=\max\{0,\lceil z\rceil\}$.
When
$K_{\infty}\not\succeq\alpha K_{\mathcal{S}}$,
the limiting temporal information remains below the required reference
level in at least one direction, and no finite $L$ satisfies the
exchange criterion.
\end{theorem}

\begin{proof}
The information remaining beyond the first $L$ observations, referred to
as the tail of the information series, is
\[
K_{\infty}-K_L
=
P^{1/2}
\left[
\sum_{k=L}^{\infty}
(A^k)^{\mathsf{T}}GA^k
\right]
P^{1/2},
\]
whose terms are positive semidefinite. Hence,
\[
\begin{aligned}
\lVert K_{\infty}-K_L\rVert_2
&\leq
\lVert P\rVert_2
\sum_{k=L}^{\infty}
\left\lVert
(A^k)^{\mathsf{T}}GA^k
\right\rVert_2
\\
&\leq
\lVert P\rVert_2
\lVert G\rVert_2
\sum_{k=L}^{\infty}\rho^{2k}
\\
&=
\frac{
\lVert P\rVert_2
\lVert G\rVert_2
\rho^{2L}
}{
1-\rho^2
},
\end{aligned}
\]
which gives the information-deficit bound in
Eq.~\ref{eq:stable-tail}.
The definition of $\gamma$ in Eq.~\ref{eq:contractive-gap} gives
precisely the positive limiting gap required by
Theorem~\ref{thm:unified-exchange}(iv). Requiring the bound in
Eq.~\ref{eq:stable-tail} to be no larger than $\gamma$ and solving the
resulting scalar inequality for $L$ gives
Eq.~\ref{eq:stable-L}. The impossibility statement follows from
Theorem~\ref{thm:unified-exchange}(iii).
\end{proof}

The result has a direct physical interpretation. In a contractive
system, later measurements contain progressively less information about
the initial state because its influence decays as the system evolves.
The temporal information therefore approaches a finite ceiling
$K_{\infty}$. The dependence on $\rho^{2L}$ shows that a rapidly
contracting system, corresponding to smaller $\rho$, reaches its finite
information ceiling over a shorter history. However, rapid contraction
can also reduce the ceiling itself, because some state directions may
decay before they are sufficiently observed.

The resulting sufficient history-length bound uses only global operator
norms and is therefore conservative, but it is broadly applicable.
Every term in the bound is directly computable from the evolution model,
prior covariance, sensor noise, and chosen spatial reference.

\subsection{Spectrally separated unitary dynamics}
\label{subsec:unitary}

Here we consider dynamics that are unitary on the active
invariant subspace, so the magnitude of each modal component is
preserved over time while its phase evolves. Lossless oscillatory or
wave systems provide a familiar example when represented in
energy-normalized modal coordinates. For such systems, repeated measurements observe different combinations of the modes as their relative phases
change. The ability of the temporal history to distinguish different
modes depends on three factors: whether each mode is sensed, whether
distinct modal phases remain separated at the sampling interval, and
how strongly the corresponding sensor signatures overlap.

Let $A$ be unitary on an $r$-dimensional invariant subspace,
\begin{equation}
A
=
V\Lambda V^{*},
\qquad
\Lambda
=
\operatorname{diag}
\left(
e^{i\theta_1},\ldots,e^{i\theta_r}
\right),
\label{eq:unitary-modal-decomposition}
\end{equation}
where the columns $v_j$ of $V$ are the modal basis vectors and $V^{*}$
denotes the conjugate transpose of $V$. A state on this subspace can be
written as
\begin{equation}
X_k
=
\sum_{j=1}^{r}
a_j e^{ik\theta_j}v_j,
\label{eq:unitary-state}
\end{equation}
where $a_j$ is the initial coefficient of mode $v_j$. The corresponding
measurement is
\begin{equation}
Y_k
=
CX_k
=
\sum_{j=1}^{r}
a_j e^{ik\theta_j}Cv_j.
\label{eq:unitary-measurement}
\end{equation}
Thus, the relative contributions of the modal sensor responses $Cv_j$
change across observation times according to the phase factors
$e^{ik\theta_j}$.

Assume that the prior shares the same modal basis,
\begin{equation}
P
=
V\operatorname{diag}(p_1,\ldots,p_r)V^{*},
\qquad
p_j>0.
\label{eq:unitary-prior}
\end{equation}
Define the noise-whitened sensor response of mode $v_j$ by
\[
b_j
=
R^{-1/2}Cv_j,
\]
and let
\begin{equation}
q_j
=
p_j\lVert b_j\rVert_2^2,
\qquad
q_{\min}
=
\min_j q_j.
\label{eq:modal-coupling}
\end{equation}
Here, $q_j$ measures the prior- and noise-weighted sensor coupling of
mode $v_j$. When every active mode is sensed, so that $q_j>0$ for all
$j$, define the modal sensor coherence by
\begin{equation}
\rho_B
=
\max_{j\neq\ell}
\frac{
\left|b_j^{*}b_\ell\right|
}{
\lVert b_j\rVert_2
\lVert b_\ell\rVert_2
},
\label{eq:modal-coherence}
\end{equation}
and the minimum sampled phase separation by
\begin{equation}
s_{\min}
=
\min_{j\neq\ell}
\left|
\sin
\left(
\frac{\theta_j-\theta_\ell}{2}
\right)
\right|.
\label{eq:sampled-phase-separation}
\end{equation}
The quantity $\rho_B$ measures the normalized similarity of the
noise-whitened sensor responses of distinct modes. The condition
$s_{\min}>0$ excludes temporal aliasing between distinct modes; that is,
it prevents two different modes from having identical sampled phase
evolution.

\begin{theorem}[Spectral-separation history bound]
\label{thm:unitary-L}
If $q_j>0$ for every $j$ and $s_{\min}>0$, then
\begin{equation}
K_L
\succeq
q_{\min}
\left[
L
-
\frac{(r-1)\rho_B}{s_{\min}}
\right]I.
\label{eq:unitary-lower}
\end{equation}
Consequently, with
\[
\kappa_{\mathcal{S}}
=
\lambda_{\max}(K_{\mathcal{S}}),
\]
where $\lambda_{\max}(K_{\mathcal{S}})$ denotes the largest eigenvalue
of the spatial-reference information operator, the sufficient condition
\begin{equation}
L
\geq
\left\lceil
\frac{(r-1)\rho_B}{s_{\min}}
+
\frac{\alpha\kappa_{\mathcal{S}}}{q_{\min}}
\right\rceil
\label{eq:unitary-L}
\end{equation}
guarantees
$K_L\succeq\alpha K_{\mathcal{S}}$.
\end{theorem}

\begin{proof}
In modal coordinates, the $(j,\ell)$ entry of $K_L$ is
\[
\sqrt{p_jp_\ell}\,
b_j^{*}b_\ell
\sum_{k=0}^{L-1}
e^{ik(\theta_\ell-\theta_j)}.
\]
For $j=\ell$, the phase factor is equal to one, and therefore
\[
(K_L)_{jj}
=
L p_j\lVert b_j\rVert_2^2
=
Lq_j.
\]
Thus, the information accumulated along each individual mode grows
linearly with $L$.

For $j\neq\ell$, the finite geometric sum satisfies
\[
\left|
\sum_{k=0}^{L-1}
e^{ik(\theta_\ell-\theta_j)}
\right|
\leq
\frac{1}{
\left|
\sin\left((\theta_j-\theta_\ell)/2\right)
\right|
}.
\]
Consequently, the normalized magnitude of each off-diagonal entry
satisfies
\[
\frac{
\left|(K_L)_{j\ell}\right|
}{
\sqrt{(K_L)_{jj}(K_L)_{\ell\ell}}
}
\leq
\frac{
\rho_B
}{
L\left|
\sin\left((\theta_j-\theta_\ell)/2\right)
\right|
}
\leq
\frac{\rho_B}{Ls_{\min}}.
\]

Define
\[
D_L
=
\operatorname{diag}
(Lq_1,\ldots,Lq_r).
\]
Since $q_j>0$ for every $j$, $D_L$ is positive definite. Applying
Gershgorin's theorem to the normalized Hermitian matrix
$D_L^{-1/2}K_LD_L^{-1/2}$, whose diagonal entries are equal to one,
gives
\[
D_L^{-1/2}K_LD_L^{-1/2}
\succeq
\left[
1
-
\frac{(r-1)\rho_B}{Ls_{\min}}
\right]I.
\]
Premultiplying and postmultiplying by $D_L^{1/2}$ gives
\[
K_L
\succeq
\left[
1
-
\frac{(r-1)\rho_B}{Ls_{\min}}
\right]D_L.
\]

When the bracketed factor is nonnegative, using
$D_L\succeq Lq_{\min}I$ gives
Eq.~\ref{eq:unitary-lower}. If the bracketed factor is negative, the
same bound follows directly from $K_L\succeq0$, because the right-hand
side of Eq.~\ref{eq:unitary-lower} is then negative semidefinite.
Therefore, Eq.~\ref{eq:unitary-lower} holds in both cases.

Since
$K_{\mathcal{S}}\preceq\kappa_{\mathcal{S}}I$,
the condition in Eq.~\ref{eq:unitary-L} ensures that the lower bound in
Eq.~\ref{eq:unitary-lower} is at least
$\alpha\kappa_{\mathcal{S}}I$, and hence
$K_L\succeq\alpha K_{\mathcal{S}}$.
\end{proof}

Equation~\ref{eq:unitary-L} relates the sufficient history length to
three observational and dynamical quantities. The value $q_{\min}$
measures the coupling of the least well-observed active mode to the
sensors after prior and noise weighting. The phase-separation parameter
$s_{\min}$ measures how distinct the sampled phase evolutions of
different modes are. The coherence $\rho_B$ measures the normalized
similarity of their noise-whitened sensor responses. Stronger modal
coupling, greater phase separation, and lower coherence reduce the
sufficient history length, whereas temporal aliasing ($s_{\min}=0$)
leaves distinct modes indistinguishable through additional temporal
measurements alone.

\subsection{Modewise dissipative dynamics}
\label{subsec:dissipative}

Dissipative systems introduce an additional effect: repeated
measurements may approach a finite information limit because an initial
mode decays before it can continue to influence the sensors. Under
simultaneous modal diagonalization, the dissipative system separates
into independent modal problems, allowing determination of an exact
exchange length rather than a conservative sufficient bound.

\begin{assumption}[Simultaneous modal diagonalization]
\label{ass:modal}
On the target subspace, suppose that $A$, $P$,
$G=C^{\mathsf{T}}R^{-1}C$, and $K_{\mathcal{S}}$ are diagonal in the
same orthonormal basis $V$:
\begin{equation}
\begin{aligned}
A
&=
V\operatorname{diag}(\lambda_j)V^{*},
\\
P
&=
V\operatorname{diag}(p_j)V^{*},
\\
G
&=
V\operatorname{diag}(g_j)V^{*},
\\
K_{\mathcal{S}}
&=
V\operatorname{diag}(s_j)V^{*},
\end{aligned}
\label{eq:simultaneous-modal-diagonalization}
\end{equation}
where
$p_j>0$, $g_j\geq0$, $s_j\geq0$, and
$|\lambda_j|\leq1$.
\end{assumption}

Here, $\lambda_j$ is the eigenvalue of $A$ associated with mode $v_j$,
and $p_j$, $g_j$, and $s_j$ are the corresponding diagonal entries of
$P$, $G$, and $K_{\mathcal{S}}$, respectively. Define
\begin{equation}
q_j:=p_jg_j,
\qquad
r_j:=|\lambda_j|.
\label{eq:modewise-parameters}
\end{equation}
Thus, $q_j$ is the prior-weighted sensor information for mode $j$, while
$r_j$ is the fraction of its amplitude retained after one time step.

\begin{theorem}[Exact modewise exchange length for dissipative dynamics]
\label{thm:modewise}
Under Assumption~\ref{ass:modal},
\begin{equation}
K_L
=
V
\operatorname{diag}
\left(
q_j
\sum_{k=0}^{L-1}r_j^{2k}
\right)
V^{*}.
\label{eq:modewise-information}
\end{equation}

For a required fraction $\alpha\in(0,1]$, a finite exchange length
exists if and only if every mode with $s_j>0$ satisfies one of the
following conditions:
\begin{enumerate}[label=\arabic*)]
    \item
    $r_j=1$ and $q_j>0$;

    \item
    $r_j=0$ and $q_j\geq\alpha s_j$;

    \item
    $0<r_j<1$ and
    \begin{equation}
    \frac{q_j}{1-r_j^2}
    >
    \alpha s_j.
    \label{eq:strict-ceiling}
    \end{equation}
\end{enumerate}

The exact minimum exchange length is
\begin{equation}
L_{\alpha}
=
\max_j \ell_j(\alpha),
\label{eq:modewise-L}
\end{equation}
where, for $s_j>0$,
\begin{equation}
\ell_j(\alpha)
=
\begin{cases}
\left\lceil
\dfrac{\alpha s_j}{q_j}
\right\rceil,
&
r_j=1,
\\[2mm]
1,
&
r_j=0
\ \text{and}\
q_j\geq\alpha s_j,
\\[2mm]
\left\lceil
\dfrac{\log(1-a_j)}
{2\log r_j}
\right\rceil_{+},
&
0<r_j<1,\quad a_j<1,
\\[3mm]
\infty,
&
\text{otherwise},
\end{cases}
\label{eq:ellj}
\end{equation}
with
\begin{equation}
a_j
:=
\frac{
\alpha s_j(1-r_j^2)
}{
q_j
}
\label{eq:modewise-aj}
\end{equation}
whenever $q_j>0$. Set $a_j=+\infty$ when $q_j=0$, and set
$\ell_j(\alpha)=0$ when $s_j=0$.

For $0<r_j<1$, equality in the information-ceiling condition,
\begin{equation}
\frac{q_j}{1-r_j^2}
=
\alpha s_j,
\label{eq:ceiling-equality}
\end{equation}
is attained only in the limit $L\rightarrow\infty$ and therefore does
not produce a finite exchange length.
\end{theorem}

\begin{proof}
Simultaneous diagonalization eliminates cross-modal terms.
Consequently, the operator inequality
$K_L\succeq\alpha K_{\mathcal{S}}$
is equivalent to the independent scalar inequalities
\begin{equation}
q_j
\sum_{k=0}^{L-1}
r_j^{2k}
\geq
\alpha s_j
\qquad
\text{for every }j.
\label{eq:modewise-scalar-condition}
\end{equation}

When $r_j=1$, the mode does not decay and
$\sum_{k=0}^{L-1}r_j^{2k}=L$. The modal condition therefore becomes
$Lq_j\geq\alpha s_j$, which gives the first case of
Eq.~\ref{eq:ellj}.

When $r_j=0$, only the first observation contains information about the
initial modal amplitude, and
\[
\sum_{k=0}^{L-1}r_j^{2k}
=
1
\qquad
\text{for every }L\geq1.
\]
A finite exchange length therefore exists exactly when
$q_j\geq\alpha s_j$, giving the second case of
Eq.~\ref{eq:ellj}.

For $0<r_j<1$, the geometric series gives
\[
q_j
\sum_{k=0}^{L-1}r_j^{2k}
=
q_j
\frac{1-r_j^{2L}}{1-r_j^2}.
\]
Hence the modal condition is
\[
q_j
\frac{1-r_j^{2L}}{1-r_j^2}
\geq
\alpha s_j,
\]
or equivalently,
\[
r_j^{2L}
\leq
1-a_j.
\]
A finite solution exists exactly when $a_j<1$, which is equivalent to
the strict ceiling condition in Eq.~\ref{eq:strict-ceiling}. Taking
logarithms then gives the third case of Eq.~\ref{eq:ellj}. Since the
full operator inequality holds exactly when every modal inequality
holds, the largest modal history requirement determines the exact
exchange length in Eq.~\ref{eq:modewise-L}.
\end{proof}

The result highlights the dual role of dissipation. Mode-dependent decay
helps distinguish temporal signatures, while the same decay limits the
total information that can be accumulated about an initial modal
amplitude. For mode $j$, the information contributed by observation
$k$ scales as $q_jr_j^{2k}$, so the accumulated information approaches
the ceiling
$q_j/(1-r_j^2)$. A finite exchange length exists only when this ceiling
strictly exceeds the required level $\alpha s_j$; in that case,
Eq.~\ref{eq:ellj} gives the exact number of observations required.

\subsection{Cyclic-advection dynamics}
\label{subsec:cyclic-advection}

Cyclic advection provides the clearest exact example of space--time
exchange. Consider a discrete representation of advection on a periodic
domain in which the sampling interval and grid spacing are chosen so
that the state advances by one grid index per time step. Let $A$
cyclically shift an $N$-point state, and let
$C=e_1^{\mathsf{T}}$ measure one fixed grid location. As the state moves
past the sensor, successive observations measure different coordinates
of the initial state. For $L\leq N$, the rows $CA^k$ are distinct
coordinate vectors.

\begin{corollary}[One sensor over one cycle equals a full spatial array]
\label{cor:cyclic}
With equal independent noise variance for every scalar observation, one
fixed sensor observed over $N$ cyclic shifts is exactly
information-equivalent to all $N$ coordinate sensors observed
simultaneously, for every positive-definite prior.
\end{corollary}

\begin{proof}
The temporal observation matrix
\[
\mathcal{O}_N
=
\begin{bmatrix}
C\\
CA\\
\vdots\\
CA^{N-1}
\end{bmatrix}
\]
is a row permutation of the spatial observation matrix $I_N$. With equal
independent noise variance for every scalar observation, their prior-
and noise-whitened Gram matrices are therefore identical, so the
temporal and spatial information operators are equal. The result follows
from Theorem~\ref{thm:gram-characterization}.
\end{proof}

\section{Stochastic state estimation and finite memory}
\label{sec:stochastic}

The deterministic results in Section~\ref{sec:linear-dynamics} describe how information about a state accumulates through measurements linked by known dynamics. This section extends the framework to stochastic dynamics, where process noise introduces new uncertainty during the measurement history and progressively reduces the relevance of older observations. Stochastically forced linear transport or shear systems
provide representative examples of this setting. The information
contained in a finite temporal window is characterized here through the conditional covariance of the state given that window. The analysis determines the limiting information available from the complete observation history and uses the rate of Riccati convergence to derive a sufficient finite window length.

Consider the stationary version of the stochastic linear model
introduced in Eq.~\ref{eq:stochastic-linear-model}. Assume that the
state has a positive-definite stationary covariance $P$ satisfying
\begin{equation}
P
=
APA^{\mathsf{T}}+Q.
\label{eq:stationary-covariance}
\end{equation}

For the length-$L$ past observation window
\[
Y_{t-L+1:t}
:=
\left(
Y_{t-L+1},\ldots,Y_t
\right),
\]
define
\begin{equation}
P_L
:=
\operatorname{Cov}
\left(
X_t\mid Y_{t-L+1:t}
\right),
\qquad
K_L
:=
P^{1/2}P_L^{-1}P^{1/2}-I.
\label{eq:finite-window-information}
\end{equation}
Here, $P_L$ is the finite-window posterior covariance based only on the
retained $L$ observations. It is distinct from the standard filtering
covariance $P_{k\mid k}$ in Eq.~\ref{eq:filtering-covariance}, which is
based on the full observation history available to the filter up to
time $k$.

\begin{proposition}[Monotonicity and infinite-past information ceiling]
\label{prop:stochastic-monotone}
Under the preceding assumptions,
\begin{equation}
P_{L+1}
\preceq
P_L
\preceq
P,
\qquad
K_{L+1}
\succeq
K_L
\succeq
0.
\label{eq:stochastic-monotonicity}
\end{equation}
Under the standard stabilizability and detectability conditions,
$P_L$ converges to the infinite-past filtering covariance $P_\infty$,
and $K_L$ converges to the corresponding information operator
\begin{equation}
P_\infty
:=
\lim_{L\rightarrow\infty}P_L,
\qquad
K_\infty
:=
P^{1/2}P_\infty^{-1}P^{1/2}-I.
\label{eq:stochastic-infinite-past}
\end{equation}
When
$K_\infty\not\succeq\alpha K_{\mathcal{S}}$,
the infinite-past information remains below the required
spatial-reference level in at least one direction, so the exchange
length is infinite:
\[
L_\alpha=\infty.
\]
\end{proposition}

\begin{proof}
The length-$(L+1)$ window $Y_{t-L:t}$ contains the length-$L$ window
$Y_{t-L+1:t}$ together with one additional observation, and conditional
covariance decreases under nested Gaussian conditioning. The information
ordering follows by inversion after prior whitening. The ceiling
statement follows from Theorem~\ref{thm:unified-exchange}(iii).
\end{proof}

Under the stabilizability and detectability conditions stated above,
$\lVert P_L-P_\infty\rVert_2$ decays exponentially with the window
length $L$ \cite{HagerHorowitz1976,Bougerol1993}. This implies that the
influence of observations omitted from the finite window decreases
exponentially as $L$ increases. This exponential-forgetting property
justifies approximating the complete observation history by a finite
window and underlies related finite-memory methods for learning linear
dynamical systems \cite{Kozdoba2019}.

\begin{theorem}[Finite-history exchange from Riccati convergence]
\label{thm:stochastic-L}
Assume
\begin{equation}
\lVert P_L-P_\infty\rVert_2
\leq
c_P\chi^L,
\qquad
c_P>0,
\qquad
0<\chi<1,
\label{eq:riccati-envelope}
\end{equation}
where $c_P$ sets the magnitude of the convergence bound and $\chi$ is
the forgetting factor controlling its decay with $L$. Let
\begin{equation}
p_*
:=
\lambda_{\min}(P_\infty)
>0
\label{eq:stochastic-pstar}
\end{equation}
denote the smallest eigenvalue of the limiting covariance. Then
\begin{equation}
\lVert K_\infty-K_L\rVert_2
\leq
\frac{
\lVert P\rVert_2 c_P
}{
p_*^2
}
\chi^L.
\label{eq:K-convergence}
\end{equation}

Fix $\alpha\in(0,1]$. If
\begin{equation}
\gamma
=
\lambda_{\min}
\left(
K_\infty-\alpha K_{\mathcal{S}}
\right)
>0,
\label{eq:stochastic-gap}
\end{equation}
then
\begin{equation}
L
\geq
\left\lceil
\frac{
\log\left(
\frac{
\lVert P\rVert_2c_P
}{
\gamma p_*^2
}
\right)
}{
\log(1/\chi)
}
\right\rceil_{+}
\label{eq:stochastic-L}
\end{equation}
is sufficient to guarantee
$K_L\succeq\alpha K_{\mathcal{S}}$.
\end{theorem}

\begin{proof}
Since
$P_L\succeq P_\infty\succeq p_*I$,
\[
\begin{aligned}
K_\infty-K_L
&=
P^{1/2}
\left(
P_\infty^{-1}-P_L^{-1}
\right)
P^{1/2}
\\
&=
P^{1/2}
P_\infty^{-1}
\left(
P_L-P_\infty
\right)
P_L^{-1}
P^{1/2}.
\end{aligned}
\]
Because
$\lVert P_\infty^{-1}\rVert_2\leq1/p_*$
and
$\lVert P_L^{-1}\rVert_2\leq1/p_*$,
taking norms and applying
Eq.~\ref{eq:riccati-envelope} gives
Eq.~\ref{eq:K-convergence}.

The definition of $\gamma$ in Eq.~\ref{eq:stochastic-gap} gives the
positive limiting gap required by
Theorem~\ref{thm:unified-exchange}(iv). Requiring the bound in
Eq.~\ref{eq:K-convergence} to be no larger than $\gamma$ and solving
the resulting scalar inequality for $L$ gives
Eq.~\ref{eq:stochastic-L}.
\end{proof}

The finite-memory exchange length depends on the limiting margin
$\gamma$ and the forgetting factor $\chi$, together with the
covariance-convergence prefactor appearing in
Eq.~\ref{eq:stochastic-L}. A larger limiting margin $\gamma$ or a
smaller forgetting factor $\chi$ reduces the required temporal window.
The forgetting factor measures how rapidly the influence of observations
preceding the retained window decays in the current-state estimate. It
plays a role analogous to the contraction or modal-decay factors in the
deterministic results, but additionally accounts for the continual
introduction of process uncertainty.

\section{Nonlinear systems: global and local exchange}
\label{sec:nonlinear}

For nonlinear systems, space--time information exchange has two
complementary aspects. Global identifiability concerns whether an entire
admissible state set can be distinguished from the observations. Local
information exchange concerns how sensitively the observations respond
to small perturbations of the state in the presence of measurement
noise. The first property is described by injectivity of the nonlinear
observation map, whereas the second is quantified by its
Fisher-information matrix.

This section applies the delay-map and Fisher-information concepts
introduced in Section~\ref{sec:background} to the comparison between
temporal and simultaneous spatial sensing. Let
$\mathcal{M}\subset\mathbb{R}^{r}$ be a compact set of admissible states,
let
\[
\Phi_{\mathcal{T},L}:\mathcal{M}\rightarrow\mathbb{R}^{mL}
\]
denote the length-$L$ temporal delay map, and let
\[
s:\mathcal{M}\rightarrow\mathbb{R}^{p}
\]
denote the simultaneous spatial-reference map. Assume that both maps
are continuous.

\begin{theorem}[Global noiseless substitution and equivalence]
\label{thm:global-nonlinear}
If $\Phi_{\mathcal{T},L}$ is injective on $\mathcal{M}$, then there
exists a continuous map
\begin{equation}
\psi_{L}:
\Phi_{\mathcal{T},L}(\mathcal{M})
\rightarrow
s(\mathcal{M})
\label{eq:global-substitution-map}
\end{equation}
such that
\begin{equation}
s(x)
=
\psi_{L}\bigl(\Phi_{\mathcal{T},L}(x)\bigr)
\qquad
\text{for every }x\in\mathcal{M}.
\label{eq:global-substitution}
\end{equation}
Thus, the noiseless temporal history determines the simultaneous
spatial-reference measurement over the entire admissible state set.

If $s$ is also injective on $\mathcal{M}$, then both
$\Phi_{\mathcal{T},L}$ and $s$ are homeomorphisms onto their images.
Consequently, if $X$ is a random state supported on $\mathcal{M}$, then
for every finite-valued measurable quantization $\mathcal{Q}(X)$,
writing $\mathrm{H}$ for Shannon entropy,
\begin{equation}
\mathrm{H}\bigl(
\mathcal{Q}(X)\mid\Phi_{\mathcal{T},L}(X)
\bigr)
=
\mathrm{H}\bigl(
\mathcal{Q}(X)\mid s(X)
\bigr)
=
0,
\label{eq:global-entropy-equivalence}
\end{equation}
and hence
\begin{equation}
I\bigl(
\mathcal{Q}(X);\Phi_{\mathcal{T},L}(X)
\bigr)
=
I\bigl(
\mathcal{Q}(X);s(X)
\bigr)
=
\mathrm{H}\bigl(\mathcal{Q}(X)\bigr).
\label{eq:global-mi-equivalence}
\end{equation}
Therefore, when both maps are injective, the temporal and
spatial-reference observations contain the same complete
finite-resolution information about the admissible state.
\end{theorem}

\begin{proof}
Because $\mathcal{M}$ is compact and $\Phi_{\mathcal{T},L}$ is
continuous and injective, it is a homeomorphism from $\mathcal{M}$ onto
$\Phi_{\mathcal{T},L}(\mathcal{M})$. Hence,
\[
\psi_{L}
=
s\circ\Phi_{\mathcal{T},L}^{-1}
\]
is continuous and satisfies
\[
s
=
\psi_{L}\circ\Phi_{\mathcal{T},L}.
\]
If $s$ is also injective, the same argument shows that it is a
homeomorphism onto its image, so either observation determines $X$, and
therefore any finite quantization $\mathcal{Q}(X)$. The Shannon-entropy and mutual-information identities then follow.
\end{proof}

The preceding result concerns global noiseless reconstruction. The
effect of measurement noise is compared locally through tangent-space
Fisher-information operators: the temporal operator introduced in
Section~\ref{sec:background} and its spatial-reference analogue. These
are
\begin{equation}
\begin{aligned}
\mathfrak{g}_{\mathcal{T},L}(x)
&=
D\Phi_{\mathcal{T},L}(x)^{\mathsf{T}}
R_{\mathcal{T},L}^{-1}
D\Phi_{\mathcal{T},L}(x),\\
\mathfrak{g}_{\mathcal{S}}(x)
&=
Ds(x)^{\mathsf{T}}
R_{\mathcal{S}}^{-1}
Ds(x),
\end{aligned}
\label{eq:nonlinear-fisher-operators}
\end{equation}
with the derivatives understood as restricted to the tangent space when
$\mathcal{M}$ is a differentiable manifold.

For the following result, assume additionally that $\mathcal{M}$ is a
compact differentiable manifold and that the temporal and spatial
observation maps are continuously differentiable on a neighborhood of
$\mathcal{M}$.

\begin{theorem}[Uniform local nonlinear exchange]
\label{thm:nonlinear-exchange}
Fix $\alpha\in(0,1]$. Suppose that, on each tangent space
$T_x\mathcal{M}$,
\begin{equation}
\mathfrak{g}_{\mathcal{T},L}(x)
\preceq
\mathfrak{g}_{\mathcal{T},L+1}(x)
\longrightarrow
\mathfrak{g}_{\mathcal{T},\infty}(x),
\label{eq:nonlinear-fisher-convergence}
\end{equation}
and that the Fisher-information tail---the information remaining beyond
history length $L$---is uniformly bounded by
\begin{equation}
\sup_{x\in\mathcal{M}}
\left\|
\mathfrak{g}_{\mathcal{T},\infty}(x)
-
\mathfrak{g}_{\mathcal{T},L}(x)
\right\|_{2}
\leq
\phi(L),
\qquad
\phi(L)\downarrow 0.
\label{eq:nonlinear-fisher-tail}
\end{equation}
Assume further that the limiting temporal information exceeds the
required spatial-reference level by a uniform margin:
\begin{equation}
\mathfrak{g}_{\mathcal{T},\infty}(x)
-
\alpha\mathfrak{g}_{\mathcal{S}}(x)
\succeq
\gamma I_{T_x\mathcal{M}}
\qquad
\text{for every }x\in\mathcal{M},
\label{eq:nonlinear-uniform-gap}
\end{equation}
for some $\gamma>0$, where $I_{T_x\mathcal{M}}$ denotes the identity on
the tangent space. Then every $L$ satisfying
$\phi(L)\leq\gamma$ guarantees
\begin{equation}
\mathfrak{g}_{\mathcal{T},L}(x)
\succeq
\alpha\mathfrak{g}_{\mathcal{S}}(x)
\qquad
\text{for every }x\in\mathcal{M}.
\label{eq:nonlinear-local-exchange}
\end{equation}

When both Fisher-information operators are positive definite on
$T_x\mathcal{M}$, this further implies
\begin{equation}
\mathfrak{g}_{\mathcal{T},L}(x)^{-1}
\preceq
\frac{1}{\alpha}
\mathfrak{g}_{\mathcal{S}}(x)^{-1},
\label{eq:nonlinear-cr-bound}
\end{equation}
where the inverses are understood on $T_x\mathcal{M}$. Thus, under the
usual regularity and local-unbiasedness assumptions, the temporal
Cram\'er--Rao lower-bound matrix is no larger than the
spatial-reference bound scaled by $1/\alpha$.
\end{theorem}

\begin{proof}
For each $x\in\mathcal{M}$, apply Theorem~\ref{thm:unified-exchange}(iv)
on the tangent space $T_x\mathcal{M}$, with
$\mathfrak{g}_{\mathcal{T},L}(x)$ and
$\mathfrak{g}_{\mathcal{S}}(x)$ as the temporal and spatial information
operators. The uniform bound $\phi(L)\leq\gamma$ ensures that the
resulting inequality holds simultaneously for every
$x\in\mathcal{M}$. When the operators are positive definite, the stated
Cram\'er--Rao comparison follows because matrix inversion reverses the
Loewner order.
\end{proof}

Theorem~\ref{thm:global-nonlinear} determines whether the temporal
history can reproduce the spatial reference throughout the admissible
state set in the noiseless setting. Theorem~\ref{thm:nonlinear-exchange}
gives a finite-history condition ensuring that its local sensitivity to
state perturbations is uniformly comparable to that of the spatial
reference. For turbulent PDEs, these conditions are most naturally
evaluated on a compact regime-restricted manifold or in reduced tangent
coordinates, since the active state geometry may vary between regimes.

\paragraph{Contractive coordinate-observation specialization.}
Consider coordinate observations
$h_m(x)=C_mx$, where $C_m$ selects $m$ state coordinates, with
independent Gaussian measurement noise of variance
$\sigma_{\mathcal{T}}^{2}$ at each observation time. Since
\begin{equation}
D\bigl(h_m\circ F^k\bigr)(x)
=
C_mD\bigl(F^k\bigr)(x),
\label{eq:contractive-coordinate-jacobian}
\end{equation}
substitution into the delay-map Fisher-information definition in
Eq.~\ref{eq:temporal-fisher-information} gives
\begin{equation}
\mathfrak{g}_{\mathcal{T},L}^{(m)}(x)
=
\frac{1}{\sigma_{\mathcal{T}}^{2}}
\sum_{k=0}^{L-1}
D\bigl(F^k\bigr)(x)^{\mathsf{T}}
C_m^{\mathsf{T}}C_m
D\bigl(F^k\bigr)(x),
\label{eq:contractive-coordinate-fisher}v
\end{equation}
where $F^k$ denotes the $k$-fold iterate of $F$, so
$D(F^k)(x)$ is its derivative at $x$. Consequently, the Fisher information remaining beyond history
length $L$ is
\begin{equation}
\mathfrak{g}_{\mathcal{T},\infty}^{(m)}(x)
-
\mathfrak{g}_{\mathcal{T},L}^{(m)}(x)
=
\frac{1}{\sigma_{\mathcal{T}}^{2}}
\sum_{k=L}^{\infty}
D\bigl(F^k\bigr)(x)^{\mathsf{T}}
C_m^{\mathsf{T}}C_m
D\bigl(F^k\bigr)(x).
\label{eq:contractive-coordinate-fisher-tail}
\end{equation}

If $\mathcal{M}$ is forward invariant and
\begin{equation}
\sup_{x\in\mathcal{M}}
\left\|DF(x)\right\|_2
\leq q<1,
\label{eq:nonlinear-contraction}
\end{equation}
then the chain rule gives
$\|D(F^k)(x)\|_2\leq q^k$. Summing the resulting geometric series
yields
\begin{equation}
\sup_{x\in\mathcal{M}}
\left\|
\mathfrak{g}_{\mathcal{T},\infty}^{(m)}(x)
-
\mathfrak{g}_{\mathcal{T},L}^{(m)}(x)
\right\|_2
\leq
\frac{
\|C_m\|_2^2 q^{2L}
}{
\sigma_{\mathcal{T}}^2(1-q^2)
}
=: \phi_m(L).
\label{eq:contractive-coordinate-tail-bound}
\end{equation}
Thus, contractivity provides an explicit uniform Fisher-information-tail
envelope with which to apply
Theorem~\ref{thm:nonlinear-exchange}.

\section{Empirical verification of the exchange theorems}
\label{sec:verification}

With the general exchange conditions and their deterministic, stochastic,
and nonlinear specializations established, we now examine their
predictions in four numerical experiments for which the relevant
information or tangent operators can be computed directly. The two
deterministic linear cases (Cases 1 and 2) test exact directional
exchange, spectral separation, finite information ceilings, and
contractive tail bounds. The stochastic case (Case 3) tests the
directional ceiling condition, finite-memory posterior saturation, and
the bound derived from Riccati convergence. The nonlinear case (Case 4)
tests global delay-map identifiability and local Fisher-information
exchange.

Throughout these tests, we carry forward the direction-resolved
comparison defined in Eq.~\ref{eq:directional-exchange}. Agreement in total mutual information or
mean-square error can conceal a poorly informed state direction, and thus each case examines an appropriate directional quantity, such as the minimum generalized eigenvalue defining $\alpha_L$, the minimum eigenvalue of an operator margin, or the worst-case local Fisher-information factor. Additional operator-level diagnostics for all four test cases are presented in the
\hyperref[sec:supplementary-material]{Supplementary Material}.

All linear experiments use 50 independent repetitions, with each
reported finite-sample estimate based on 60,000 evaluation draws. In the
deterministic cases, an additional 60,000 state--measurement pairs are
used to fit the estimator. Each pair starts from an independently drawn
state; its measurement history is generated by propagating that state
for $L$ observation times and adding independent measurement noise. As a result, the samples represent different initial states and noise realizations, not successive points from one trajectory.

In the stochastic case, each evaluation draw is an independent
current-state Bayes residual from the exact finite-window posterior
distribution, which includes process and measurement noise. No training
set is required because the exact Bayes estimator is used. This
experiment hence provides an exact finite-sample Monte Carlo
verification of the posterior-error, information-ceiling, and
finite-memory predictions instead of an end-to-end estimator-training test.

The nonlinear case combines numerical and empirical tests. The Jacobian,
Fisher-information, and tail criteria are evaluated on 180 calibration
states and 120 independent held-out states drawn from the prescribed
state domain. The empirical collision test uses 10 independent
ensembles of 420 states to identify distinct states with similar
measurement histories. The empirical local-reconstruction test uses 18
held-out anchor states, with 50 independent measurement-noise
realizations per anchor and noise level. These tests address local
Jacobian rank, global distinguishability, directional information, and
reconstruction uncertainty separately.

Across the experiments described above, theoretical quantities are calculated
directly from the specified dynamical and measurement models. The
construction of the estimators and empirical information quantities is
described in Appendix~\ref{app:empirical-procedures}.

\subsection{Case 1: Multimodal oscillatory system}
\label{subsec:case1}

The first deterministic case uses a six-dimensional state comprising
three oscillatory mode pairs:
\begin{equation}
\begin{aligned}
x_{k+1} &= Ax_k,
\qquad
A=QBQ^{\mathsf{T}},\\
B &=
\operatorname{blkdiag}
\bigl(
R(\theta_1),R(\theta_2),R(\theta_3)
\bigr),\\
R(\theta)
&=
\begin{bmatrix}
\cos\theta & -\sin\theta\\
\sin\theta & \cos\theta
\end{bmatrix},
\qquad
(\theta_1,\theta_2,\theta_3)
=
(0.32,1.18,2.16).
\end{aligned}
\label{eq:case1-dynamics}
\end{equation}

The angles are modal rotations per observation interval. The
$6\times6$ orthonormal type-II discrete cosine transform (DCT-II)
matrix $Q$ mixes the modal coordinates into the physical coordinates,
so each coordinate sensor measures a superposition of the three mode
pairs. The modal prior variances are $1.25$, $0.95$, and $0.75$,
repeated within each pair. Temporal and simultaneous spatial
measurements have noise standard deviation $0.50$ per channel. The
spatial reference observes all six coordinates once. The temporal
designs use the one-based coordinate sets $\{1\}$, $\{1,4\}$,
$\{1,3,5\}$, and $\{1,2,4,5\}$ for $m=1,2,3,4$, respectively, with
$m$ denoting the number of sensors. The four configurations do not form
a fully nested sequence. The exchange target is
$\alpha_{\star}=0.90$.

\begin{figure}[!t]
    \centering
    \includegraphics[width=\textwidth]{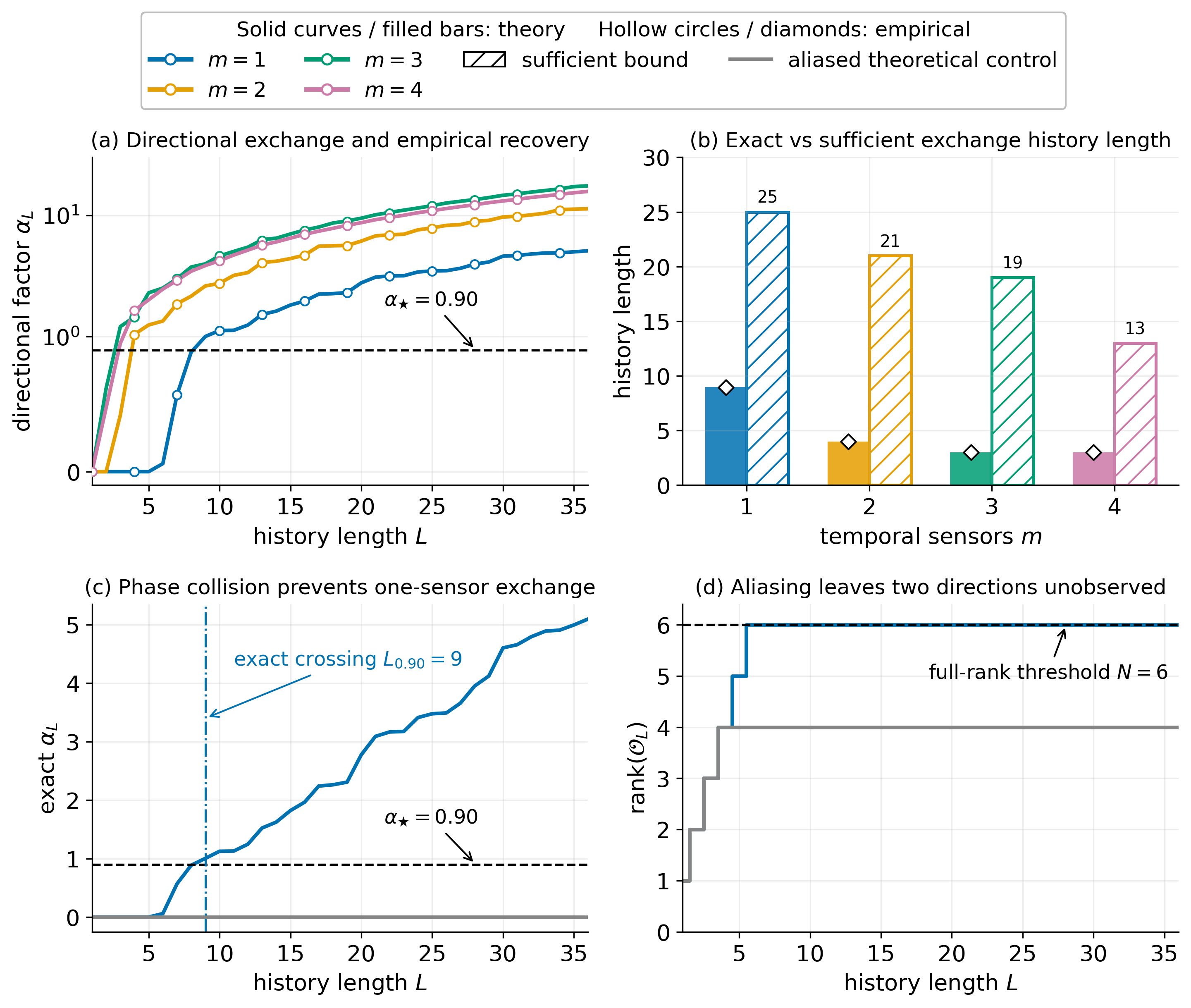}
    \caption{Case 1 results.
    (a) Exact directional substitution factors and empirical means
    $\pm1$ standard deviation across 50 repetitions; the uncertainty
    bands are too small to be visible.
    (b) Exact and empirical exchange lengths compared with the
    conservative sufficient bounds from
    Theorem~\ref{thm:unitary-L}.
    (c,d) Exact one-sensor comparisons for separated and aliased modal
    phases. Repeating a modal angle produces temporal aliasing, so the
    substitution factor remains zero and the observability rank
    saturates below the state dimension.}
    \label{fig:case1-multimodal}
\end{figure}

Figure~\ref{fig:case1-multimodal}(a) shows the theoretical substitution
factors and their empirical recovery, while
Figure~\ref{fig:case1-multimodal}(b) compares the exchange history
lengths. The exact exchange history lengths $L_{0.90}(m)$---the first
$L$ satisfying
$K_{\mathcal{T}}(L)\succeq0.90K_{\mathcal{S}}$---are
$9$, $4$, $3$, and $3$ for $m=1,2,3,4$. The corresponding empirical
first-crossing means are $8.94$, $4.00$, $3.00$, and $3.00$. Only the
one-sensor result varies across repetitions, with crossings at $L=8$ or
$9$. This agreement shows that the operator-defined exchange lengths are
recovered independently from fitted estimators.
Theorem~\ref{thm:unitary-L} gives the sufficient history lengths
$25$, $21$, $19$, and $13$. These larger values reflect the
conservative nature of the theorem: they are sufficient guarantees,
rather than point predictions of the exchange length.

At the corresponding exact exchange history lengths, the empirical
mutual information $\widehat{I}(X;Y_{\mathcal{T},L})$ ranges from
$5.421$ to $6.336$ nats across $m=1$ to $m=4$, exceeding the
Corollary~\ref{cor:posterior-consequence} lower bound of $4.499$ nats.
The empirical normalized reconstruction errors range from $0.126$ to
$0.162$, below the operator-derived upper bound of $0.219$. Thus, once
$K_{\mathcal{T}}(L)\succeq0.90K_{\mathcal{S}}$,
the temporal histories also satisfy the total-information and
reconstruction-error consequences predicted by
Corollary~\ref{cor:posterior-consequence}.

Figures~\ref{fig:case1-multimodal}(c) and
\ref{fig:case1-multimodal}(d) show a deliberately aliased comparison
for the one-sensor design. The second modal angle is set equal to the
first,
\begin{equation}
(\theta_1,\theta_2,\theta_3)
=
(0.32,0.32,2.16),
\label{eq:case1-aliased-angles}
\end{equation}
so two mode pairs have identical sampled temporal evolution. A single
sensor cannot distinguish their contributions; this loss of
distinguishability is temporal aliasing. In this case the minimum
sampled phase separation $s_{\min}$ is zero, and
Theorem~\ref{thm:unitary-L} provides no finite sufficient history.
Figure~\ref{fig:case1-multimodal}(c) shows that $\alpha_L$ remains zero
throughout the plotted range $L\leq36$, while
Figure~\ref{fig:case1-multimodal}(d) shows that the observability rank
saturates at $4/6$. By
Proposition~\ref{prop:observability-positive-information}, the temporal
information operator is thus singular. Since the spatial reference informs
all six directions, no positive exchange factor can be achieved: a
two-dimensional state subspace remains unobservable regardless of
history length. This comparison confirms that the successful exchange
in Figures~\ref{fig:case1-multimodal}(a) and
\ref{fig:case1-multimodal}(b) results from temporal separation of the
modes, not simply from using a longer input history.

\subsection{Case 2: Heterogeneous advection--diffusion--reaction transport}
\label{subsec:case2}

The second deterministic case uses a 64-dimensional state
$x_k\in\mathbb{R}^{64}$, representing concentration perturbations on an
$8\times8$ grid:
\begin{equation}
x_{k+1}=Ax_k,
\qquad
A=cRDA_{\mathrm{adv}},
\label{eq:case2-dynamics}
\end{equation}
where $A_{\mathrm{adv}}$, $D$, and $R$ represent advection, diffusion,
and reaction, respectively, and $c$ is chosen so that
$\|A\|_2=0.97<1$. The advection follows an open serpentine path from the
lower-left inlet to the upper-left outlet. At path position $q$, the
fraction
\begin{equation}
w_q
=
1-0.003
\left[
1+0.3\sin\left(\frac{2\pi q}{64}\right)
\right]
\label{eq:case2-advection-weight}
\end{equation}
moves to position $q+1$, while $1-w_q$ moves to $q+2$; material
transported beyond the final cells leaves the domain. Diffusion is
applied through $D=\exp(10^{-4}L_h)$, where $L_h$ is the no-flux
five-point grid Laplacian. The reaction operator is
$R=\operatorname{diag}(e^{-r_i})$, where $r_i=r(x_i,y_i)$ and the
spatially varying rate is
\begin{equation}
r(x,y)
=
0.002
+
0.001\sin(2\pi x)\cos(2\pi y).
\label{eq:case2-reaction-rate}
\end{equation}

Unlike the oscillatory case in Section~\ref{subsec:case1}, the sensors directly observe physical grid
cells, and information about an initial perturbation must be transported
to a sensor before being attenuated by reaction and outflow.

The propagator has spectral radius $\rho(A)=0.588$. Its non-normality is
quantified by the relative commutator measure
\begin{equation}
\frac{
\|A^{\mathsf{T}}A-AA^{\mathsf{T}}\|_F
}{
\|A\|_F^2
}
=
0.0224.
\label{eq:case2-nonnormality}
\end{equation}
The difference between $\rho(A)=0.588$ and $\|A\|_2=0.97$ is important
for interpreting the conservative norm-based bound below. Initial states
are drawn from a zero-mean Gaussian distribution with covariance
\begin{equation}
P_{ij}
=
\exp
\left(
-\frac{\|\boldsymbol{\xi}_i-\boldsymbol{\xi}_j\|_2}{0.24}
\right)
+
0.05\delta_{ij},
\label{eq:case2-prior}
\end{equation}
where $\boldsymbol{\xi}_i$ and $\boldsymbol{\xi}_j$ are the normalized
grid-coordinate vectors. Temporal sensor noise has standard deviation
$0.25$ per channel, whereas the simultaneous 64-coordinate spatial
reference has standard deviation $0.80$. The temporal designs use
$m=1,2,3,4$ cells distributed along the transport path. Their one-based
grid coordinates are
\begin{equation}
\begin{aligned}
m=1:\quad
&\{(1,8)\},\\
m=2:\quad
&\{(1,4),(1,8)\},\\
m=3:\quad
&\{(5,3),(7,6),(1,8)\},\\
m=4:\quad
&\{(1,2),(1,4),(1,6),(1,8)\}.
\end{aligned}
\label{eq:case2-sensor-locations}
\end{equation}
The target is again $\alpha_{\star}=0.90$.

\begin{figure}[!t]
    \centering
    \includegraphics[width=\textwidth]{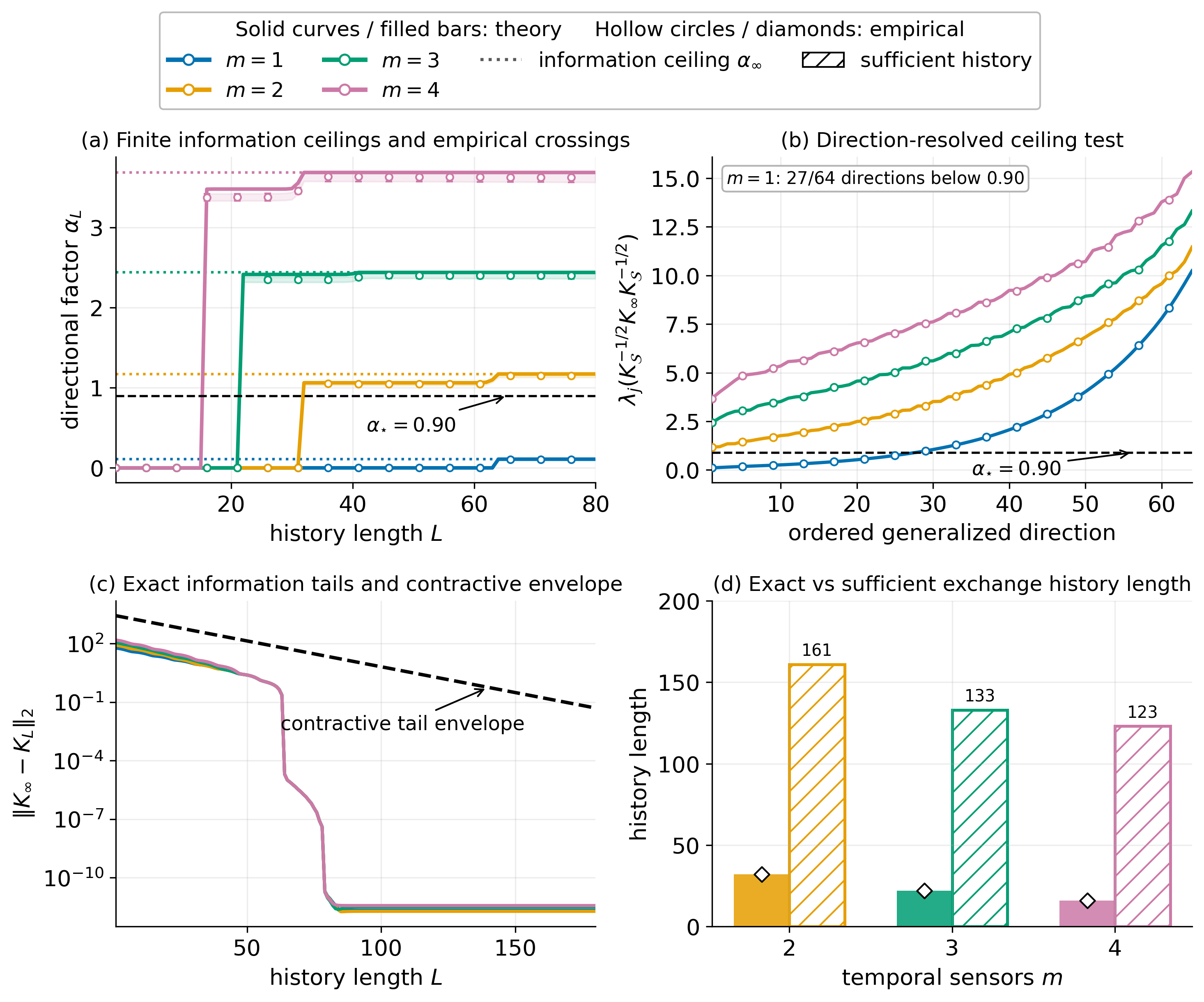}
    \caption{Case 2 results.
    (a) Exact finite-history substitution factors, theoretical
    information ceilings, and empirical means $\pm1$ standard deviation.
    (b) Theoretical generalized ceiling spectra and empirical $L=80$ spectra, including the 27 directions below the target for the one-sensor design.
    (c) Exact information deficits and the contractive upper envelope
    from Theorem~\ref{thm:stable-tail}.
    (d) Exact and empirical exchange lengths compared with the
    conservative sufficient bounds from
    Theorem~\ref{thm:stable-tail}.}
    \label{fig:case2-transport}
\end{figure}

Figure~\ref{fig:case2-transport}(a) shows the finite-history substitution
factors and their limiting ceilings. For $m=1$, the observability matrix
first reaches full rank at $L=64$, but the information ceiling is only
$\alpha_{\infty}=0.11<\alpha_{\star}$. The empirical value does not
reach the target at any history length. Although every state direction
eventually affects the measurement history, as required by classical
observability, the weakest direction never receives enough information
to match $90\%$ of the spatial reference. This verifies the
finite-ceiling impossibility condition of
Theorem~\ref{thm:unified-exchange}(iii) and shows that full observability
alone is insufficient for quantitative information exchange.

Figure~\ref{fig:case2-transport}(b) resolves this failure direction by
direction. For the one-sensor design, 27 of the 64 generalized ceiling
eigenvalues remain below $0.90$, even though its limiting total mutual
information exceeds that of the spatial reference. The scalar
information total is therefore misleading: strong information in some
directions compensates numerically for severe deficits in others. The
generalized spectrum exposes precisely the failure identified by the
paper's positive-semidefinite exchange criterion.

For $m=2,3,4$, the ceilings are $1.171$, $2.438$, and $3.685$, which
are larger than $\alpha_{\star}$; thus
Theorem~\ref{thm:unified-exchange} classifies the target as attainable.
The exact exchange history lengths are $32$, $22$, and $16$,
respectively, and all 50 empirical repetitions cross at those same
lengths. Distributing sensors along the path reduces the maximum
distance that an initial perturbation must travel before being observed,
producing progressively shorter exchange histories.

Figure~\ref{fig:case2-transport}(c) tests the contractive tail inequality
in Theorem~\ref{thm:stable-tail}. For every sensor configuration and
evaluated history length, the exact deficit
$\|K_{\infty}-K_L\|_2$ remains below the theoretical envelope. The
maximum ratios of the exact deficit to the bound are $0.022$, $0.031$,
$0.042$, and $0.057$ for $m=1,2,3,4$.
Figure~\ref{fig:case2-transport}(d) compares the exact and empirical
crossings with the sufficient history lengths from
Theorem~\ref{thm:stable-tail}. The certified lengths are $161$, $133$,
and $123$ for $m=2,3,4$, substantially larger than the exact values
$32$, $22$, and $16$. This conservatism is expected because the
certificate uses the global contraction $\|A\|_2=0.97$, whereas the
eventual decay is much faster and is governed by $\rho(A)=0.59$. No
sufficient history exists for $m=1$ because its information ceiling
lies below the target.

\subsection{Case 3: Stochastic non-normal shear flow}
\label{subsec:case3}

The stochastic state
$x_k=[u_k^{\mathsf{T}},v_k^{\mathsf{T}}]^{\mathsf{T}}
\in\mathbb{R}^{96}$
contains the streamwise streak velocity $u$ and cross-stream roll
velocity $v$ on an $8\times6$ grid:
\begin{equation}
x_{k+1}
=
Ax_k+w_k,
\qquad
w_k\sim\mathcal{N}(0,0.004I_{96}),
\label{eq:case3-stochastic-model}
\end{equation}
where $A=\exp(\Delta t F)$, $\Delta t=0.18$, and
\begin{equation}
F
=
\begin{bmatrix}
\mathcal{L} & -sI\\
0 & \mathcal{L}
\end{bmatrix},
\qquad
\mathcal{L}
=
-\operatorname{diag}(U)D_x+\nu\Delta-\mu I.
\label{eq:case3-generator}
\end{equation}
Here $U(y)=y$ is the plane-Couette profile, $D_x$ is the streamwise
derivative, and $\Delta$ is the two-dimensional diffusion operator. The
parameters are $\nu=0.015$, $\mu=0.12$, and $s=2.20$. The off-diagonal
block represents the lift-up mechanism by which cross-stream rolls
generate streamwise streaks. It also makes the dynamics non-normal:
$A$ is asymptotically stable, but its non-orthogonal modes can produce
substantial transient amplification. The stationary prior covariance
satisfies
\begin{equation}
P
=
APA^{\mathsf{T}}
+
0.004I_{96}.
\label{eq:case3-stationary-covariance}
\end{equation}
Each point sensor measures the local streak velocity, with independent
measurement noise of standard deviation $0.20$. The nested temporal
arrays contain $m=2,4,6,8,$ or $12$ sensors, and the simultaneous
spatial reference observes all 12 locations at the current time. The
target is $\alpha_{\star}=0.90$. Although $\rho(A)=0.97<1$, the
one-step norm is $\|A\|_2=1.18$, and the maximum energy gain
$\|A^k\|_2^2$ reaches $27.05$ at $k=35$. The test therefore examines
the stochastic finite-memory theory under strong transient
amplification.

For each sensor count $m$ and history length $L$, the observations are
stacked as $Y_{t-L+1:t}^{(m)}$. Their joint covariance with $X_t$ is
constructed from the stationary lag relation
\begin{equation}
\operatorname{Cov}(X_t,X_{t-d})
=
A^dP,
\label{eq:case3-lag-covariance}
\end{equation}
with measurement-noise covariance $R_m$ added to the diagonal
observation blocks. Writing
\begin{equation}
\begin{aligned}
\Sigma_{XY,L}^{(m)}
&=
\operatorname{Cov}
\bigl(X_t,Y_{t-L+1:t}^{(m)}\bigr),\\
\Sigma_{YY,L}^{(m)}
&=
\operatorname{Cov}
\bigl(Y_{t-L+1:t}^{(m)}\bigr),
\end{aligned}
\label{eq:case3-window-covariances}
\end{equation}
Gaussian conditioning gives
\begin{equation}
\begin{aligned}
P_L^{(m)}
&=
P
-
\Sigma_{XY,L}^{(m)}
\bigl(\Sigma_{YY,L}^{(m)}\bigr)^{-1}
\bigl(\Sigma_{XY,L}^{(m)}\bigr)^{\mathsf{T}},\\
K_L^{(m)}
&=
P^{1/2}
\bigl(P_L^{(m)}\bigr)^{-1}
P^{1/2}
-
I.
\end{aligned}
\label{eq:case3-window-information}
\end{equation}
These matrices generate the exact finite-window results in
Figure~\ref{fig:case3-stochastic-shear}.

\begin{figure}[!t]
    \centering
    \includegraphics[width=\textwidth]{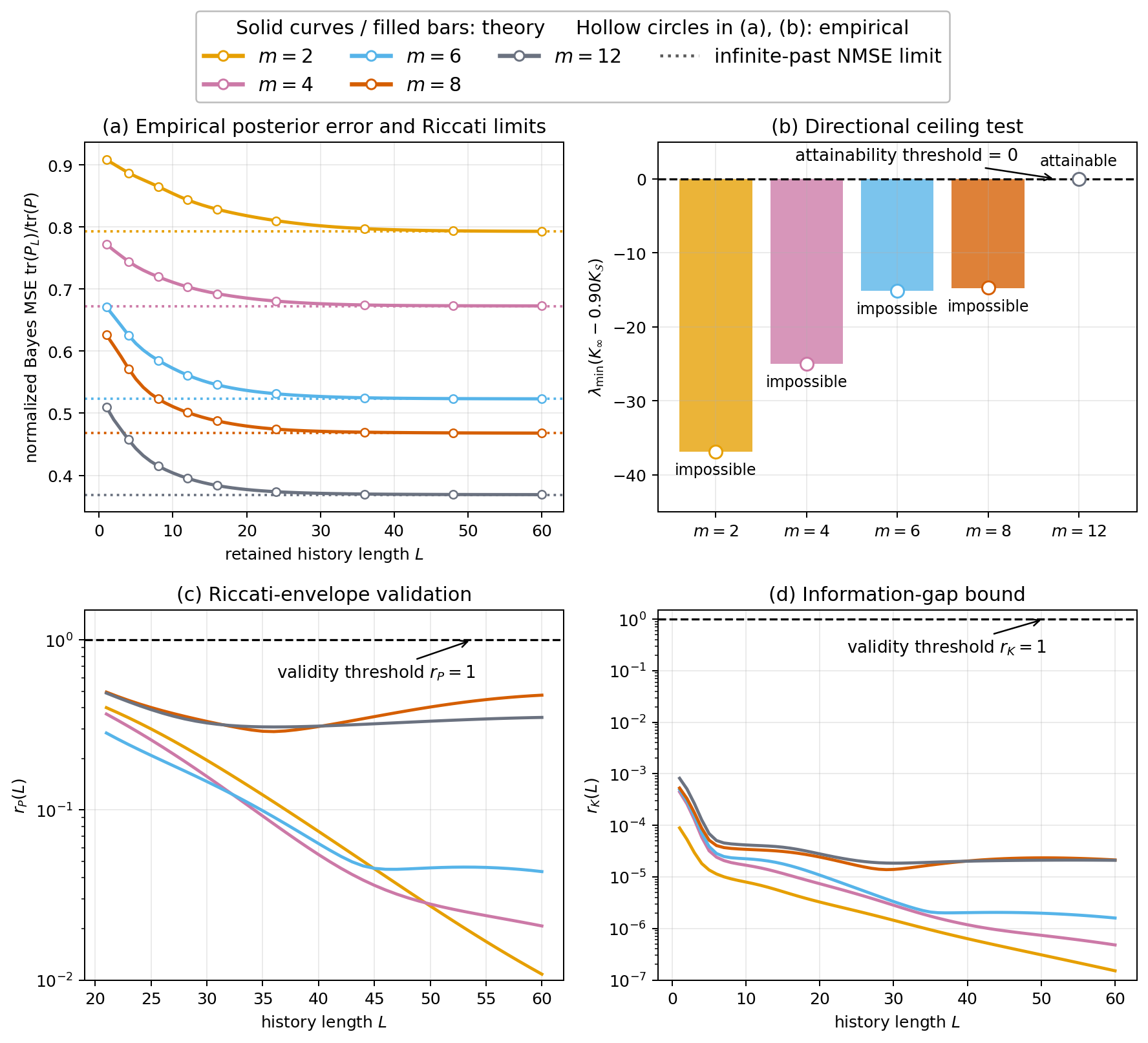}
    \caption{Case 3 results.
    (a) Exact normalized posterior errors, empirical Monte Carlo means
    $\pm1$ standard deviation, and infinite-past Riccati limits.
    (b) Exact infinite-past directional margins and empirical $L=60$ estimates; zero is the attainability threshold.
    (c) Held-out covariance-envelope ratios.
    (d) Exact information gaps divided by the corresponding
    Theorem~\ref{thm:stochastic-L} bounds. Panels (c) and (d) contain
    exact model-based quantities rather than Monte Carlo estimates.}
    \label{fig:case3-stochastic-shear}
\end{figure}

Figure~\ref{fig:case3-stochastic-shear}(a) shows close agreement between
the exact finite-window posterior errors and the Monte Carlo estimates
for every sensor count. The limiting normalized errors are $0.79$,
$0.67$, $0.52$, $0.47$, and $0.37$ for $m=2,4,6,8,$ and $12$,
respectively. These nonzero limits confirm the finite-information
saturation predicted by
Proposition~\ref{prop:stochastic-monotone}: process noise continually
introduces uncertainty and reduces the relevance of remote observations.
A separate matrix audit finds no violations, above numerical tolerance,
of the monotonicity and limiting-order relations stated in the
proposition. The verification therefore applies to the full posterior
and information matrices, not only to their traces.

Figure~\ref{fig:case3-stochastic-shear}(b) tests the $90\%$
spatial-reference target using the weakest-direction margin from
Proposition~\ref{prop:stochastic-monotone}. The exact margins are
$-36.87$, $-25.01$, $-15.14$, $-14.74$, and $0.00$, with $10$, $8$,
$6$, $4$, and $0$ failing directions, respectively. Thus, the sparse
arrays reduce posterior error but cannot reach the required
spatial-reference information in every direction, regardless of history
length, except for the $m=12$ design, which is attainable because its
current measurement already contains the complete reference array; its
nonnegative margin is zero to three-decimal precision.

Figures~\ref{fig:case3-stochastic-shear}(c) and
\ref{fig:case3-stochastic-shear}(d) test the covariance and information
bounds associated with Theorem~\ref{thm:stochastic-L}. The geometric
covariance envelope is calibrated over $L=5,\ldots,20$, using a $2\%$
safety factor, and validated independently over $L=21,\ldots,60$. The
maximum held-out covariance-envelope ratio is $0.493$, and the maximum
information-gap/bound ratio is $8.18\times10^{-4}$. Both remain below
the validity threshold of one for every sensor configuration,
confirming the premise and consequence of
Theorem~\ref{thm:stochastic-L}.

For an information gap satisfying
\begin{equation}
\left\|
K_{\infty}^{(m)}
-
K_L^{(m)}
\right\|_2
\leq
0.05
\left\|
K_{\infty}^{(m)}
\right\|_2,
\label{eq:case3-finite-memory-criterion}
\end{equation}
the exact finite-memory length is $L=11$ for $m=2$ and $L=5$ for
$m\geq4$. These values measure convergence to each design's own
infinite-history information ceiling $K_{\infty}^{(m)}$, not exchange
with $K_{\mathcal{S}}$.

\subsection{Case 4: Nonlinear saturated coupled-map lattice}
\label{subsec:case4}

The nonlinear state $x_k\in\mathbb{R}^{24}$ evolves according to
\begin{equation}
x_{k+1}
=
F(x_k)
=
\rho U\tanh(\eta W x_k),
\label{eq:case4-dynamics}
\end{equation}
where $U$ and $W$ are fixed orthogonal mixing matrices and the
hyperbolic tangent acts componentwise. The parameters $\rho=0.68$ and
$\eta=1.35$ give the global derivative bound
\begin{equation}
\|DF(x)\|_2
\leq
q
:=
\rho\eta
=
0.918
<
1.
\label{eq:case4-contraction}
\end{equation}
Hence the admissible ball
$\mathcal{M}=\{x\in\mathbb{R}^{24}:\|x\|_2\leq1.5\}$ is forward
invariant, and the influence of remote observations decays
geometrically. The nonlinear saturation makes the delay Jacobian state
dependent, while the two mixing matrices couple all state directions.

The spatial reference observes all 24 state coordinates simultaneously
with noise standard deviation $\sigma_{\mathcal{S}}=0.50$. The temporal
designs measure nested coordinate sets with
$\sigma_{\mathcal{T}}=0.12$: $m=2$ uses coordinates $(1,13)$;
$m=4$ adds $(7,19)$; $m=6$ adds $(4,10)$; and $m=8$ adds
$(16,22)$. The required local substitution level is
$\alpha_{\star}=0.80$.

\begin{figure}[!t]
    \centering
    \includegraphics[width=\textwidth]{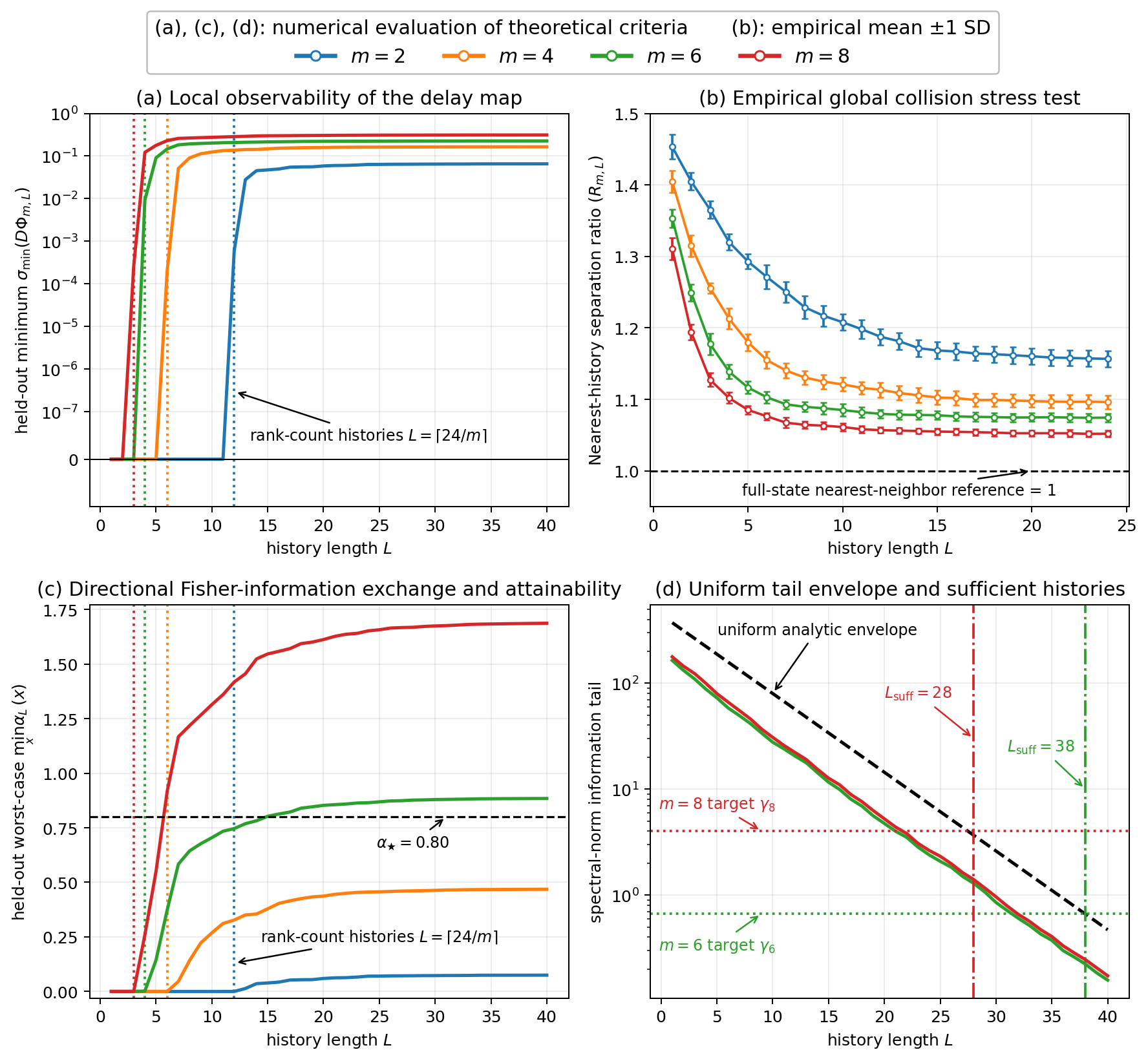}
    \caption{Case 4 results.
    (a) Minimum singular value of the exact delay Jacobian over
    independent held-out states; dotted vertical lines show the
    necessary rank-count histories.
    (b) Empirical nearest-history state separation relative to the
    full-state nearest-neighbor reference, reported as mean
    $\pm1$ standard deviation across independent state clouds.
    (c) Minimum held-out directional substitution factors and the
    target level.
    (d) Exact sampled information tails, sampled limiting margins,
    and the uniform analytic tail envelope used to obtain sufficient
    histories.}
    \label{fig:case4-nonlinear-exchange}
\end{figure}

Theorem~\ref{thm:global-nonlinear} requires the delay map to be globally
injective. Full column rank of its Jacobian establishes only local
observability, so Figure~\ref{fig:case4-nonlinear-exchange} examines
these properties separately. In
Figure~\ref{fig:case4-nonlinear-exchange}(a), the minimum
delay-Jacobian singular value becomes positive at the rank-count limits
$L=12$, $6$, $4$, and $3$ for $m=2$, $4$, $6$, and $8$,
respectively. Thus the necessary count $mL\geq24$ is attained without
additional rank loss on the held-out states.

This local rank result does not exclude a global collision, in which two
distinct states generate identical delay histories,
\begin{equation}
x_i\neq x_j,
\qquad
\Phi_{m,L}(x_i)
=
\Phi_{m,L}(x_j),
\label{eq:case4-global-collision}
\end{equation}
or a near-collision, for which their delay histories are very similar.
Figure~\ref{fig:case4-nonlinear-exchange}(b) therefore performs an
independent near-collision stress test. For each independent set of 420
sampled states, let $j_{\Phi}(i)$ denote the distinct state whose delay
history is closest to that of $x_i$, and let $j_x(i)$ denote the
distinct state closest to $x_i$ in the original state space. The
nearest-history separation ratio is defined as
\begin{equation}
R_{m,L}
=
\frac{
\operatorname{median}_i
\|x_i-x_{j_{\Phi}(i)}\|_2
}{
\operatorname{median}_i
\|x_i-x_{j_x(i)}\|_2
}.
\label{eq:case4-separation-ratio}
\end{equation}
The numerator is the median state-space separation between states
selected as nearest neighbors from their delay histories. The
denominator is the median nearest-neighbor separation obtainable when
the full state vectors are directly available. Because $j_x(i)$
minimizes the state-space distance, $R_{m,L}\geq1$. A value close to
one indicates that delay histories distinguish the sampled states
almost as well as direct access to the full states, whereas a larger
value indicates that similar histories can correspond to more widely
separated states.

Figure~\ref{fig:case4-nonlinear-exchange}(b) reports the mean and one
standard deviation of $R_{m,L}$ across 10 independently sampled sets of
420 states. At $L=24$, the mean ratios are $1.157$, $1.096$, $1.074$,
and $1.052$ for $m=2$, $4$, $6$, and $8$, respectively. The observed
decrease toward one as $m$ and $L$ increase, together with the positive
minimum singular values of the delay Jacobian, provides evidence that
delay histories effectively distinguish typical states within the
sampled sets. This finite-sample result supports, but does not prove,
global injectivity over the entire continuous state domain.

Figure~\ref{fig:case4-nonlinear-exchange}(c) tests the directional
Fisher-information condition in
Theorem~\ref{thm:nonlinear-exchange}. The minimum sampled
infinite-history substitution factors are $0.063$, $0.439$, $0.967$,
and $1.803$. The $m=2$ and $m=4$ designs therefore fail the $0.80$
target at explicitly sampled states, despite satisfying the rank count.
This demonstrates that local observability is necessary but
insufficient for quantitative information exchange. For $m=6$ and
$m=8$, the held-out worst-case curves first cross the target at $L=15$
and $L=6$, respectively.

Figure~\ref{fig:case4-nonlinear-exchange}(d) compares the largest
Fisher-information tails over the calibration states with the uniform
contraction envelope $\phi_m(L)$ used to apply
Theorem~\ref{thm:nonlinear-exchange}. Combining this envelope with the
sampled limiting margins gives sufficient-length estimates of $L=38$
for $m=6$ and $L=28$ for $m=8$. These exceed the direct held-out
crossings because the analytic envelope does not use the detailed
trajectory or sensor geometry. Although the tail envelope is uniform
over $\mathcal{M}$, the limiting margins are established only on the
calibration states. The reported lengths are therefore
sample-conditioned numerical estimates, not proofs of uniform exchange
over the entire admissible ball.

\begin{figure}[!t]
    \centering
    \includegraphics[width=\textwidth]{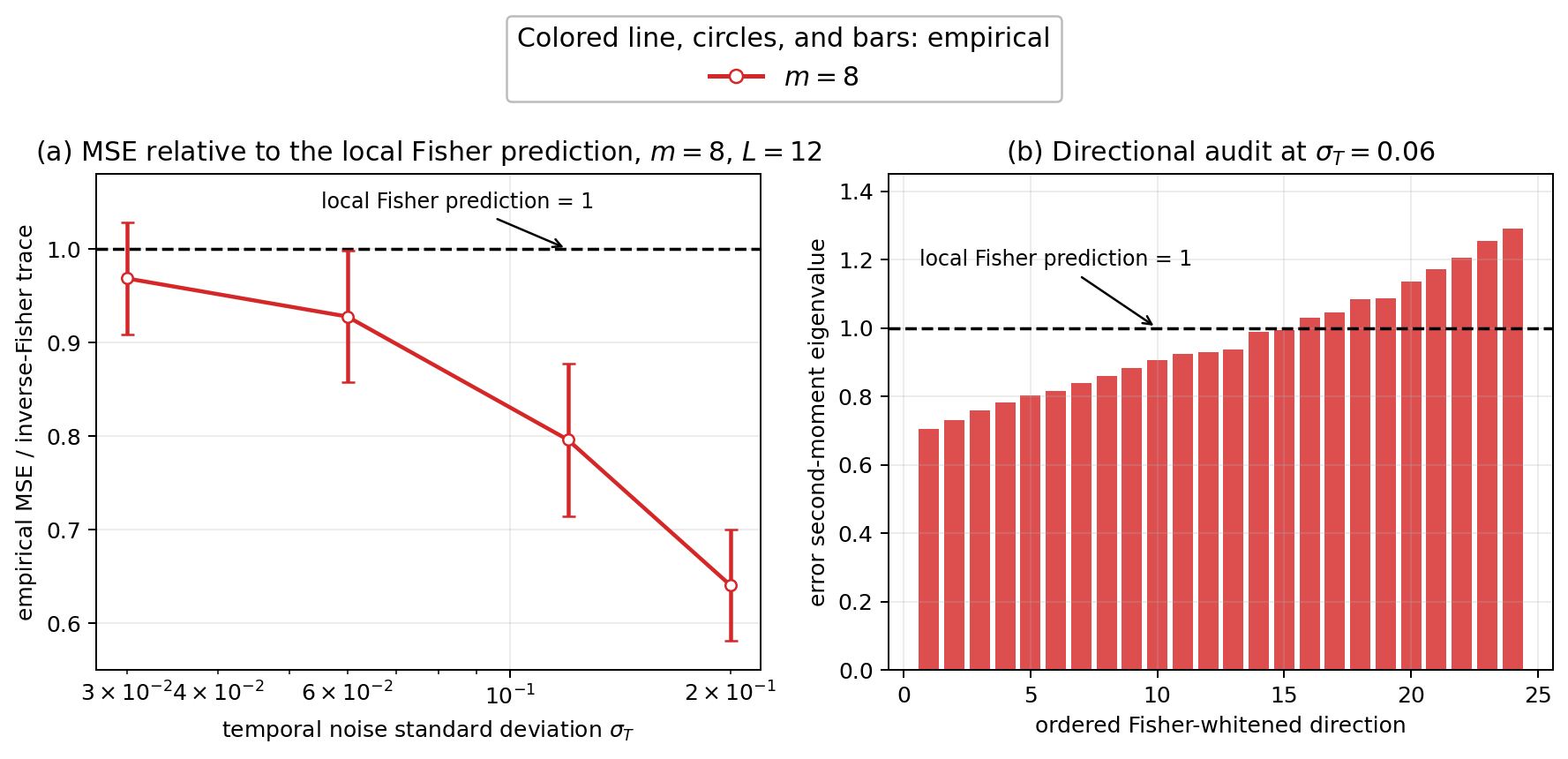}
    \caption{Independent local Fisher verification for the eight-sensor
    design at history length 12.
    (a) Empirical reconstruction MSE relative to the inverse-Fisher
    trace over four noise levels.
    (b) Eigenvalues of the pooled Fisher-whitened empirical error
    second-moment matrix at temporal-noise standard deviation $0.06$.}
    \label{fig:case4-fisher-validation}
\end{figure}

Figure~\ref{fig:case4-fisher-validation} tests a statistical consequence
of the local Fisher theory using independent noisy histories and a
Gauss--Newton estimator. The test uses $m=8$, $L=12$, and 18 additional
held-out anchor states that are not used in the preceding calculations.
The exchange calculations in Figure~\ref{fig:case4-nonlinear-exchange}
use $\sigma_{\mathcal{T}}=0.12$ as specified above; for this separate
local-reconstruction test, $\sigma_{\mathcal{T}}$ is varied over four
noise levels. At $\sigma_{\mathcal{T}}=0.03$ and $0.06$, the empirical
MSE divided by the inverse-Fisher trace is approximately $0.97$ and
$0.93$, respectively; the one-standard-deviation intervals reach or
closely approach the small-noise prediction of one. At larger noise
levels, projection onto $\mathcal{M}$, nonlinear estimator bias, and
departures from the local convergence basin move the reconstruction
outside the approximately unbiased small-noise regime. These deviations
delimit the regime in which the local Fisher approximation is accurate
rather than contradicting the theory.

The directional test in
Figure~\ref{fig:case4-fisher-validation}(b) provides a stronger check
than the scalar MSE ratio. At $\sigma_{\mathcal{T}}=0.06$, the
eigenvalues of the Fisher-whitened empirical error second-moment matrix
range from $0.70$ to $1.29$ around the asymptotic theoretical value of
one.

Across the four cases, the supplementary comparisons verify agreement between the empirical and theoretical information operators and posterior covariances, identify the physical structure of the weakest-informed directions, and test
the nonlinear Fisher prediction at the matrix level. These checks complement the main results across this Section by showing that the observed agreement extends across individual state directions.

\section{Conclusion}
\label{sec:conclusion}

This work establishes a quantitative, direction-resolved framework for
determining when a sparse temporal measurement history can substitute for a
specified instantaneous spatial sensor array. Its principal contribution is to
formulate space--time information exchange through the operator inequality
$K_{\mathcal{T}}(L)\succeq\alpha K_{\mathcal{S}}$ of
Eq.~\ref{eq:directional-exchange}, which requires the temporal measurement
history to retain at least a prescribed fraction $\alpha$ of the
spatial-reference information in every state direction.
Theorem~\ref{thm:unified-exchange} formalizes this comparison by bringing prior variability, sensor geometry,
measurement and process noise, dynamical space--time coupling, and temporal
sampling into a single criterion. The framework thus complements
observability, dynamical sampling, Bayesian experimental design, and
finite-memory filtering by connecting these ideas directly to posterior
uncertainty and the temporal context required for full-state reconstruction.

The central conclusion from this work is that space--time information interchangeability is a
directional and reference-dependent property. A temporal history may match a
spatial array in total information or average reconstruction error while
remaining substantially less informative in a particular state direction.
This distinction separates three questions that are often conflated: whether
the measurements distinguish the relevant states, whether the limiting
temporal information is sufficient relative to the spatial reference, and how
much history is required when that information level is attainable. Increasing
the history length cannot overcome an inadequate information ceiling.
Attainability should therefore be assessed before seeking a finite exchange
length.

The specialized results show how the governing dynamics determine this
exchange. Modal mixing, spectral separation, and transport can expose state
components that are not measured directly, whereas temporal aliasing can leave
distinct components indistinguishable. Diffusion, reaction, dissipation, and
outflow can weaken transported information before it reaches a sensor. The
dynamics thus play a dual role: they can reveal otherwise inaccessible spatial
information while also limiting how much of it survives. In stochastic
systems, continual process noise produces a finite information ceiling and
causes the influence of remote observations to decay. In nonlinear systems,
global delay-map distinguishability and local Fisher information provide
complementary requirements: distinct states must generate distinct histories,
and nearby states must remain sufficiently separated under measurement noise.

The numerical experiments support these conclusions at both scalar and
operator levels. The deterministic experiments recover the predicted
directional exchange behavior and demonstrate failure under temporal aliasing
or an insufficient limiting information level. The transport case links the
weakest-informed directions to physical regions whose influence is attenuated
before reaching the sensors. The stochastic experiment confirms posterior
saturation and finite-memory convergence while distinguishing convergence to a
temporal design's own ceiling from successful exchange with the spatial
reference. The nonlinear experiment shows that sufficiently informative delay
histories attain full Jacobian rank and separate the sampled states, while
their directional Fisher information approaches its limiting value and
accurately predicts reconstruction uncertainty in the small-noise regime. The
matrix-resolved comparisons further demonstrate agreement between the
theoretical and empirical results in their directional and componentwise
structure.

We would like to emphasize that the conclusions remain subject to the assumptions of each setting. The
linear-Gaussian comparisons are exact, whereas the nonlinear Fisher results are
local and most reliable in the small-noise regime. Sampled collision tests
provide evidence of global distinguishability but do not prove injectivity
over a continuous state domain, and sampled nonlinear margins should be
interpreted as numerical certificates over the evaluated states. With these notes in mind, the framework provides a practical basis for designing
temporal sensing systems in plasma diagnostics, fluid and climate monitoring,
structural sensing, and other high-dimensional dynamical applications.
Temporal history can reduce spatial sensor coverage only when the dynamics
expose the required state directions with sufficient strength before noise,
dissipation, and unresolved forcing erase their information. The theory
developed here makes this principle quantitatively testable.

A practical consequence is that measurement sufficiency can be separated from
state-estimator performance. The information operators characterize what the
temporal measurements allow one to infer under the specified dynamics, sensor
configuration, noise statistics, and prior state distribution. If they show
that the information available from the temporal sensor array is insufficient,
unsuccessful reconstruction reflects a fundamental limitation of the
measurement information rather than the reconstruction algorithm. If the
measurements are theoretically sufficient but the achieved state-reconstruction
error exceeds the corresponding posterior-error limit, the shortfall reflects
the state-estimator model's inability to extract all available information.
Low average reconstruction error, however, does not establish that every
dynamically relevant direction is adequately constrained. For SHRED and
related machine-learning-based decoder architectures, the framework can
hence guide sensor placement and history-length selection, diagnose whether
reconstruction failure is information- or estimator-limited, and provide an
information-limited reference against which decoder error can be assessed.

\section*{Data and code availability}

All data used in this study are generated synthetically from the dynamical and measurement models specified in the manuscript. The codes used to generate the data and required to reproduce the numerical experiments will be deposited in a public repository upon publication.

\appendix

\section{Empirical estimation and verification procedures}
\label{app:empirical-procedures}

For each deterministic linear sensing design
$a\in\{\mathcal{S},\mathcal{T}\}$, representing either a simultaneous
spatial array or a temporal history of length $L$, samples
$\{X^{(i)},Y_a^{(i)}\}_{i=1}^{N}$ are generated from the specified
dynamical and measurement models. The samples are divided into
independent training and test sets. A state-estimator model
\begin{equation}
f_a:Y_a\mapsto\widehat{X}_a
\label{eq:empirical-estimator}
\end{equation}
is fitted separately for the temporal and reference spatial sensor
arrangements using the training measurement--state pairs. For each
design, the estimator parameters are determined by minimizing the
empirical mean-square reconstruction error,
\begin{equation}
\widehat{\theta}_a
=
\arg\min_{\theta}
\frac{1}{N_{\mathrm{train}}}
\sum_{i=1}^{N_{\mathrm{train}}}
\left\|
X^{(i)}
-
f_{\theta}\bigl(Y_a^{(i)}\bigr)
\right\|_2^2.
\label{eq:empirical-estimator-training}
\end{equation}

For each test sample, define the reconstruction error
\begin{equation}
E_a^{(i)}
=
X^{(i)}
-
f_a\bigl(Y_a^{(i)}\bigr).
\label{eq:empirical-reconstruction-error}
\end{equation}
The empirical reconstruction-error matrix is
\begin{equation}
\widehat{M}_a
=
\frac{1}{N_{\mathrm{test}}}
\sum_{i=1}^{N_{\mathrm{test}}}
E_a^{(i)}E_a^{(i)\mathsf{T}}.
\label{eq:empirical-error-matrix}
\end{equation}
Strictly, $\widehat{M}_a$ is an error second-moment matrix rather than a
centered covariance matrix. This definition includes any systematic
state-estimator bias. When the fitted estimator approximates the
conditional mean,
\begin{equation}
f_a(Y_a)
\approx
\mathbb{E}[X\mid Y_a],
\label{eq:empirical-conditional-mean}
\end{equation}
the empirical error matrix estimates the Bayes posterior-error matrix.
In the jointly Gaussian setting,
\begin{equation}
\widehat{M}_a
\approx
P_a,
\label{eq:empirical-posterior-approximation}
\end{equation}
where $P_a$ is the theoretical posterior covariance defined in
Eq.~\ref{eq:design-posterior-covariance}.

Let $\widehat{P}_a$ denote the covariance used for whitening. In the two
deterministic cases, this is the held-out empirical estimate of the prior
covariance $P$. In the stochastic case, the exact stationary covariance
$P$ defined in Eq.~\ref{eq:case3-stationary-covariance} is used instead.
The raw empirical information operator is
\begin{equation}
\widehat{K}^{\,\mathrm{raw}}_a
=
\widehat{P}_a^{1/2}
\widehat{M}_a^{-1}
\widehat{P}_a^{1/2}
-
I.
\label{eq:empirical-information-operator-raw}
\end{equation}
For the numerical comparisons,
$\widehat{K}^{\,\mathrm{raw}}_a$ is symmetrized and projected onto the
positive-semidefinite cone. Eigenvalues not exceeding
$10^{-10}\max\{1,\|\widehat{K}^{\,\mathrm{raw}}_a\|_2\}$ are set to
zero. The resulting matrix is denoted by $\widehat{K}_a$.

For the jointly Gaussian experiments, the empirical mutual-information
estimate used in the numerical comparisons is
\begin{equation}
\widehat{I}(X;Y_a)
=
\frac{1}{2}
\log\det\bigl(I+\widehat{K}_a\bigr).
\label{eq:empirical-mutual-information}
\end{equation}
Before the numerical positive-semidefinite projection, the corresponding
covariance identity is
\[
\frac{1}{2}
\log
\frac{\det\widehat{P}_a}{\det\widehat{M}_a}.
\]
This scalar quantity measures total uncertainty reduction, while the
information operator retains its directional distribution. As shown in
Proposition~\ref{prop:scalar-mutual-information}, mutual information
alone is insufficient for directional exchange.

The empirical temporal substitution factor is calculated from the
estimated temporal and spatial operators. When
$\widehat{K}_{\mathcal{S}}\succ0$,
\begin{equation}
\widehat{\alpha}_L
=
\lambda_{\min}
\left(
\widehat{K}_{\mathcal{S}}^{-1/2}
\widehat{K}_{\mathcal{T}}(L)
\widehat{K}_{\mathcal{S}}^{-1/2}
\right).
\label{eq:empirical-substitution-factor}
\end{equation}
If $\widehat{K}_{\mathcal{S}}$ is singular, the inverse-square-root
formula is not used. Instead,
\begin{equation}
\widehat{\alpha}_L
=
\sup
\left\{
\alpha\geq0:
\widehat{K}_{\mathcal{T}}(L)
-
\alpha\widehat{K}_{\mathcal{S}}
\succeq0
\right\}.
\label{eq:empirical-substitution-factor-singular}
\end{equation}
For a prescribed target $\alpha_{\star}$, this is tested directly
through
\begin{equation}
\lambda_{\min}
\left(
\widehat{K}_{\mathcal{T}}(L)
-
\alpha_{\star}\widehat{K}_{\mathcal{S}}
\right)
\geq0.
\label{eq:empirical-margin-test}
\end{equation}
The singular-reference comparison is performed in the full state space,
rather than only on
$\operatorname{range}(\widehat{K}_{\mathcal{S}})$, so that mixed
directions involving its null space are retained. This full-space
margin test is used in Case 3.

For a prescribed exchange level $\alpha_{\star}$, the empirical
exchange length is
\begin{equation}
\widehat{L}_{\alpha_{\star}}
=
\min
\left\{
L:
\widehat{\alpha}_L\geq\alpha_{\star}
\right\}.
\label{eq:empirical-exchange-length}
\end{equation}

When the theoretical temporal operator has a finite information ceiling
$K_{\mathcal{T},\infty}$, its empirical counterpart is approximated by
increasing $L$ until $\widehat{K}_{\mathcal{T}}(L)$ and
$\widehat{I}(X;Y_{\mathcal{T},L})$ reach a statistically stable plateau.
The resulting estimates are denoted by
$\widehat{K}_{\mathcal{T},\infty}$ and
$\widehat{I}_{\mathcal{T},\infty}$.

The empirical quantities
$\widehat{M}_a$, $\widehat{K}_a$, $\widehat{\alpha}_L$, and
$\widehat{L}_{\alpha_{\star}}$ are compared with the corresponding
theoretical posterior covariances, information operators, substitution
factors, and exchange lengths,
$P_a$, $K_a$, $\alpha_L$, and $L_{\alpha_{\star}}$, respectively,
which are introduced and developed in
Sections~\ref{sec:information-operators}--\ref{sec:stochastic}.
Agreement is assessed over several independent Monte Carlo repetitions
and reported with sampling uncertainty. Directional quantities are also
reported because agreement in trace, mutual information, or total
mean-square error does not guarantee agreement in the least-informed
state direction.

\paragraph{Estimator choice and interpretation.}
In the two deterministic cases, a separate multivariate ridge regression
(linear regression with a small regularization term that improves
numerical stability) is fitted for every sensor configuration and
history length. The ridge coefficient is $10^{-8}$ times the mean
diagonal entry of the centered training-measurement Gram matrix. For
each temporal sensor configuration, the coefficient obtained for the
maximum history is used for all shorter history prefixes.

Case 1 generates the training and test samples explicitly. In Case 2,
the equivalent Gaussian sample statistics are generated directly to
avoid materializing the much larger state--history sample arrays. For a
jointly Gaussian vector, the sample mean and centered scatter matrix
have independent normal and Wishart distributions, respectively.
Drawing these sufficient statistics and fitting the ridge estimator from
them is algebraically equivalent to fitting and evaluating it on the
corresponding explicit Gaussian samples.

In the stochastic case, the finite-window Gaussian posterior is
calculated directly. Its finite-sample variability is generated from the Wishart distribution,
which gives the exact sampling distribution of the zero-mean Gaussian
residual second-moment matrix used here.

In the nonlinear experiment, the Jacobian-rank, collision,
directional-information, and tail tests do not require an estimator. The
separate local reconstruction test uses a Gauss--Newton estimator, an
iterative nonlinear least-squares method that finds the state whose
predicted measurement history most closely matches the observed
history. Although this estimator is not generally identical to the
nonlinear conditional mean, it provides a direct test of the local
Fisher prediction in the small-noise regime. The reconstruction result
is interpreted as a local statistical check, while the principal
nonlinear verification is provided by the Jacobian, collision, and
information-operator tests.

\paragraph{Nonlinear verification route.}
At each calibration or held-out state, the temporal delay map and its
Jacobian are evaluated directly. The smallest Jacobian singular value
measures local rank; nearest-history searches over independent state
clouds test sampled global distinguishability; and the temporal and
spatial Fisher operators are formed from the corresponding Jacobians
and noise covariances to calculate directional substitution factors and
information tails. The separate Gauss--Newton reconstructions produce
empirical error second-moment matrices for comparison with the
inverse-Fisher prediction in the small-noise regime.

\clearpage
\phantomsection
\label{sec:supplementary-material}
\section*{Supplementary Material: Matrix and directional comparisons}

\setcounter{figure}{0}
\renewcommand{\thefigure}{S\arabic{figure}}
\setcounter{equation}{0}
\renewcommand{\theequation}{S\arabic{equation}}

This Supplement reports full-matrix and weakest-direction diagnostics
supporting the comparisons in Section~\ref{sec:verification}. Where
information-operator residuals are shown, the plotted matrix is
\begin{equation}
\Delta K_{\mathcal{T}}(L)
=
\frac{
\widehat{K}_{\mathcal{T}}(L)
-
K_{\mathcal{T}}(L)
}{
\left\|K_{\mathcal{T}}(L)\right\|_{F}
},
\label{eq:supp-information-residual}
\end{equation}
and the reported relative error is
\begin{equation}
e_K
=
\frac{
\left\|
\widehat{K}_{\mathcal{T}}(L)
-
K_{\mathcal{T}}(L)
\right\|_{F}
}{
\left\|K_{\mathcal{T}}(L)\right\|_{F}
}.
\label{eq:supp-information-relative-error}
\end{equation}
The corresponding posterior-error-matrix relative error is
\begin{equation}
e_M
=
\frac{
\left\|
\widehat{M}_{\mathcal{T}}(L)
-
P_{\mathcal{T}}(L)
\right\|_{F}
}{
\left\|P_{\mathcal{T}}(L)\right\|_{F}
},
\label{eq:supp-posterior-relative-error}
\end{equation}
where $\widehat{M}_{\mathcal{T}}(L)$ is the empirical error
second-moment matrix and $P_{\mathcal{T}}(L)$ is the corresponding
theoretical posterior covariance, consistent with the definitions in
Appendix~\ref{app:empirical-procedures}.

\subsection*{Case 1: Multimodal information-operator comparison}

The residuals in Figure~\ref{fig:supp-case1-residuals} show no dominant
unresolved coordinate direction. Together with the small errors in
$\alpha_L$, this confirms that the empirical exchange crossings reproduce
the full information-operator comparison. The corresponding
posterior-matrix relative errors $e_M$ are $0.137\%$, $0.140\%$,
$0.166\%$, and $0.142\%$ for $m=1,2,3,$ and $4$, respectively. The
exchange crossings are therefore not artifacts of compensating
directional errors concealed by agreement in mutual information or
normalized MSE.

\begin{figure}[!htbp]
    \centering
    \includegraphics[width=\textwidth]
    {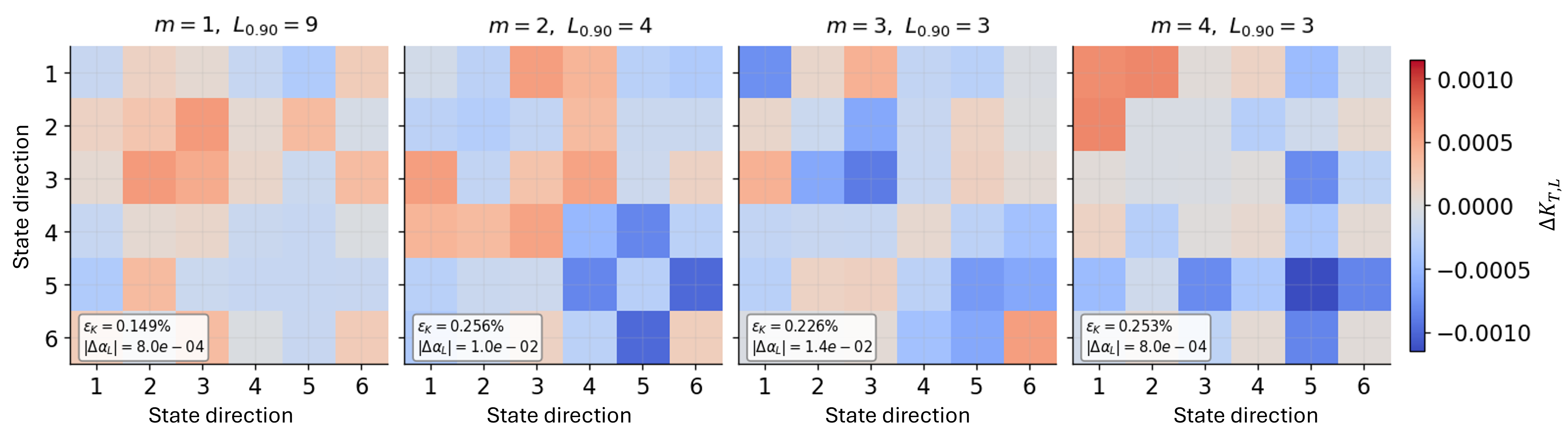}
    \caption{Normalized temporal information-operator residuals at the
    exact exchange length for each temporal sensor count in Case 1.
    The inset reports the relative operator error and the absolute
    substitution-factor error; the corresponding posterior-matrix
    errors are reported in the text.}
    \label{fig:supp-case1-residuals}
\end{figure}

\subsection*{Case 2: Transport weakest directions and full-matrix comparison}

For each sensor count $m$, let $v_{\min}^{(m)}$ denote the generalized
eigenvector associated with the smallest ceiling substitution factor
$\alpha_{\infty}^{(m)}$:
\begin{equation}
K_{\infty}^{(m)}v_{\min}^{(m)}
=
\alpha_{\infty}^{(m)}
K_{\mathcal{S}}v_{\min}^{(m)}.
\label{eq:supp-case2-weakest-direction}
\end{equation}
This vector represents the state-perturbation direction for which the
limiting temporal measurements provide the least information relative
to the spatial reference. Figure~\ref{fig:supp-case2-weakest-directions}
shows the spatial distribution of the normalized weakest-direction
amplitude. These patterns give a physical interpretation of the
directional information ceiling. Perturbations originating far from a
sensor must travel farther before being measured and are weakened by
diffusion, reaction, and outflow during transport. These perturbations
form the least-informed state directions. Adding sensors reduces the
distance from each part of the domain to an observation point, confines
the weakest directions to shorter intervals between sensors, and
improves the limiting information in those directions.

\begin{figure}[!htbp]
    \centering
    \includegraphics[width=\textwidth]
    {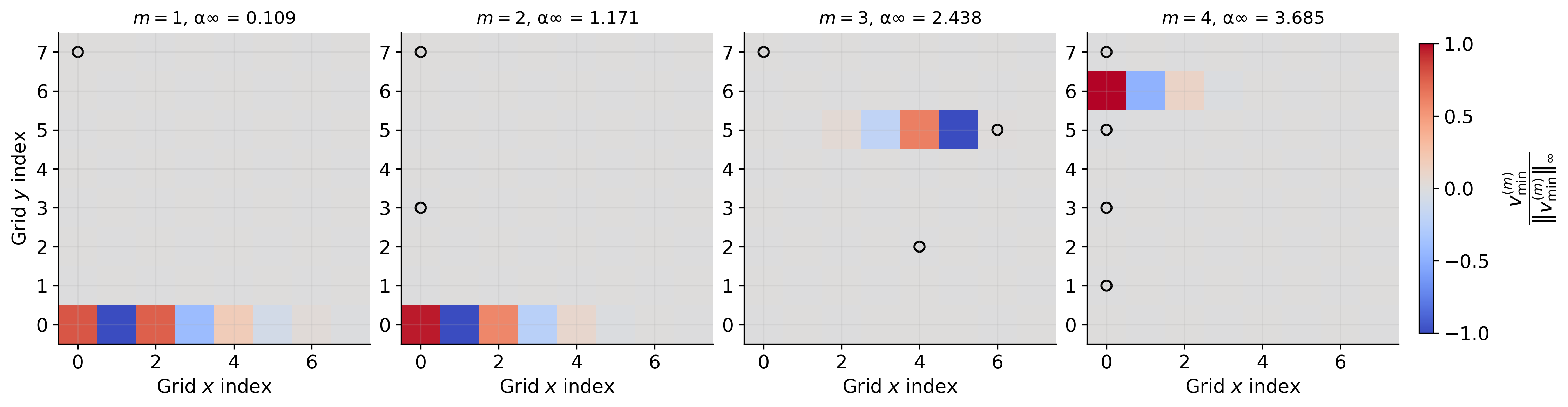}
    \caption{Physical-space representations of the generalized
    eigenvectors associated with the smallest limiting substitution
    factor in Case 2. Hollow circles mark coordinates observed by each
    temporal sensor array; strong red or blue values identify where the
    weakest direction is concentrated.}
    \label{fig:supp-case2-weakest-directions}
\end{figure}

The small full-matrix residuals shown in
Figure~\ref{fig:supp-case2-residuals} confirm that both the one-sensor
impossibility result and the exchange crossings for $m\geq2$ arise from
the predicted directional information geometry rather than estimation
error. Figures~\ref{fig:supp-case2-weakest-directions} and
\ref{fig:supp-case2-residuals} thus support the weakest-direction and
limiting-ceiling criteria underlying the transport results.

\begin{figure}[!htbp]
    \centering
    \includegraphics[width=\textwidth]
    {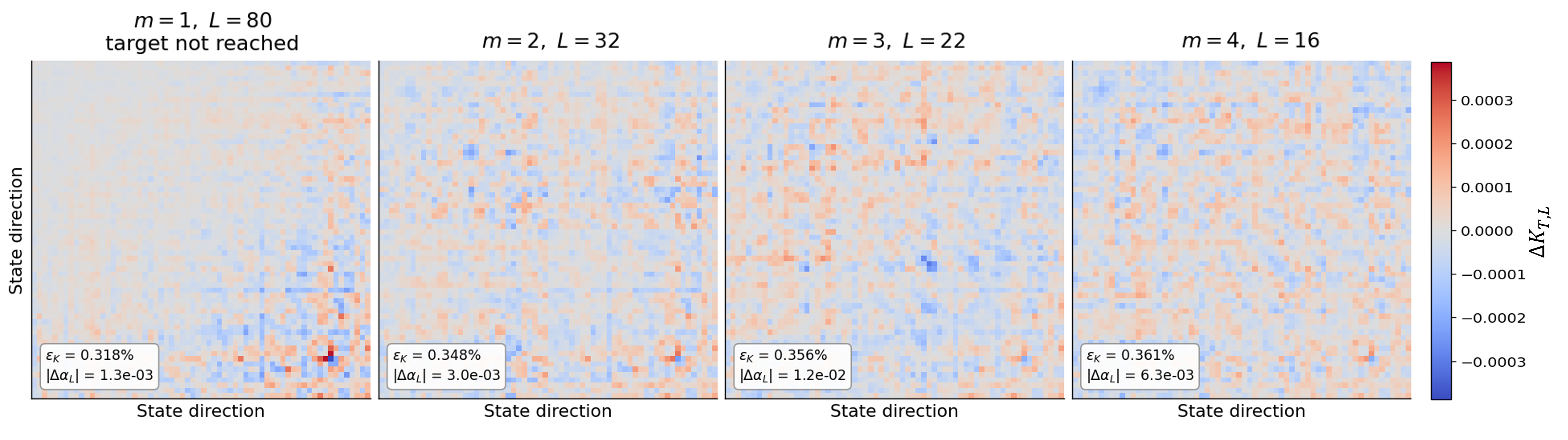}
    \caption{Normalized temporal information-operator residuals for
    Case 2. The one-sensor design is evaluated at the maximum history
    shown because it never reaches the target; the other designs are
    evaluated at their exact exchange lengths. Insets report relative
    operator and substitution-factor errors.}
    \label{fig:supp-case2-residuals}
\end{figure}

\subsection*{Case 3: Stochastic full-operator comparison}

The full-matrix comparison in Figure~\ref{fig:supp-case3-operator}
confirms that the Monte Carlo calculation reproduces the
direction-resolved finite-window information and posterior covariance,
rather than only their normalized traces. It supports the full-operator
calculations underlying the stochastic information ceiling and
finite-memory results.

\begin{figure}[!htbp]
    \centering
    \includegraphics[width=\textwidth]
    {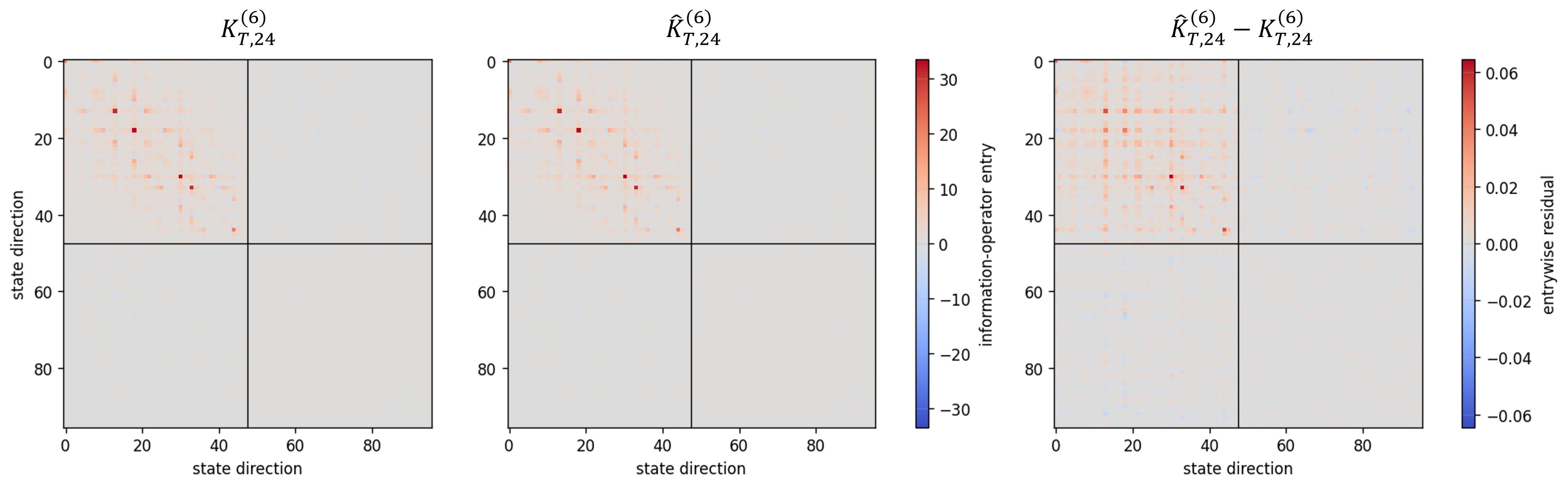}
    \caption{Theoretical information operator, empirical mean, and
    entrywise residual for the stochastic six-sensor design at history
    length $24$. Horizontal and vertical lines separate the
    streamwise-streak and cross-stream-roll components, revealing the
    four physical information blocks.}
    \label{fig:supp-case3-operator}
\end{figure}

\subsection*{Case 4: Nonlinear Fisher-whitened error comparison}

The $16.8\%$ difference in Figure~\ref{fig:supp-case4-fisher} closely
matches the expected $16.7\%$ finite-sample scale, while the residual
shows no obvious structured off-diagonal pattern. This is consistent
with the small-noise error prediction of the local Fisher theory and
provides no evidence of a systematic directional discrepancy. This
test does not, however, verify the uniform nonlinear exchange condition
of Theorem~\ref{thm:nonlinear-exchange}, which requires a positive
limiting information margin at every state in the admissible set
$\mathcal{M}$.

\begin{figure}[!htbp]
    \centering
    \includegraphics[width=\textwidth]
    {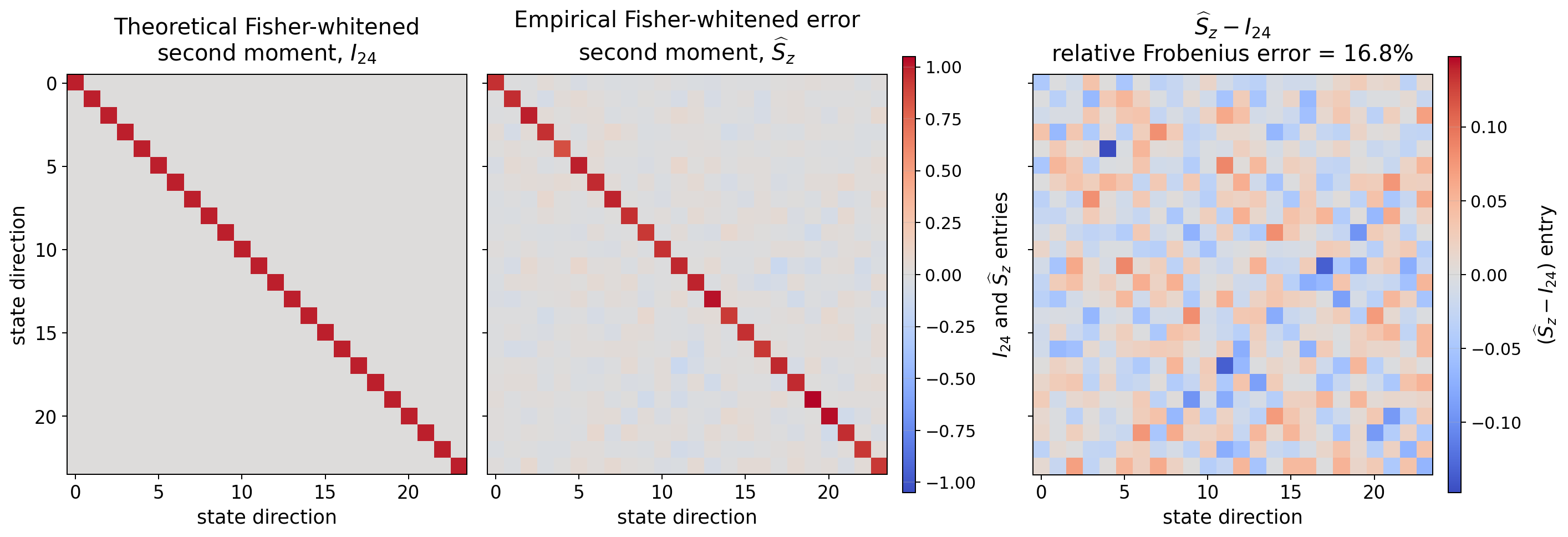}
    \caption{Theoretical local Fisher prediction, empirical
    Fisher-whitened error second-moment matrix, and their residual for
    the nonlinear eight-sensor design at history length $12$ and
    temporal-noise standard deviation $0.06$. The observed relative
    Frobenius difference is consistent with the expected finite-sample
    scale.}
    \label{fig:supp-case4-fisher}
\end{figure}
\FloatBarrier

\bibliographystyle{unsrtnat}
\bibliography{references}

\end{document}